\documentclass[sigconf]{acmart}
\usepackage{popets}
\usepackage{graphicx, subcaption}
\usepackage{lipsum}
\usepackage{mathtools}
\usepackage{cuted}
\usepackage{algorithm,algpseudocode}

\usepackage{amsmath}

\newtheorem{theorem}{Theorem}
\usepackage{blkarray}

\usepackage[utf8]{inputenc} 
\usepackage[T1]{fontenc}    
\usepackage{hyperref}       
\usepackage{url}            
\usepackage{amsfonts}       
\usepackage{dsfont}
\usepackage{bbm}
\usepackage{bm}
\usepackage{xparse}
\usepackage{amsfonts}
\usepackage{textcomp}
\usepackage{amsthm}
\usepackage{xcolor}
\usepackage{multirow}
\newtheorem{lemma}[theorem]{Lemma}

\allowdisplaybreaks

\theoremstyle{plain}
\newtheorem{thm}{Theorem}
\newtheorem*{theorem*}{Theorem}

\newtheorem{defin}{Definition}
\newtheorem{assump}{Assumption}

\usepackage{array}

\usepackage{mathtools}
\usepackage{cuted}
\newcommand{\mcal}{\mathcal}
\newcommand{\mb}{\mathbf}

\setcopyright{popets}
\copyrightyear{2024}

\acmYear{2024}
\acmConference{Arxiv Preprint}
\begin{document}

\title[Differentially Private Federated Learning without Noise Addition: When is it Possible?]{Differentially Private Federated Learning\\
without Noise Addition: When is it Possible?}

\author{Jiang Zhang, Konstantinos Psounis, Salman Avestimehr}

\affiliation{
  \institution{ University of Southern California  }
  \city{  Los Angeles  }
  \country{  United States  }}
\email{   }

\renewcommand{\shortauthors}{Zhang et al.}

\begin{abstract}
Federated Learning (FL) with Secure Aggregation (SA) has gained significant attention as a privacy preserving framework for training machine learning models 
while preventing the server from learning information about users' data from their individual encrypted model updates.
Recent research has extended privacy guarantees of FL with SA by bounding the information leakage through the aggregate model over multiple training rounds thanks to leveraging the ``noise" from other users' updates. However, the privacy metric used in that work (mutual information) measures the on-average privacy leakage, without providing any privacy guarantees for worse-case scenarios. 
To address this, in this work we study the conditions under which FL with SA can provide worst-case differential privacy guarantees. Specifically, we formally identify the necessary condition that SA can provide DP without addition noise. We then prove that when the randomness inside the aggregated model update is Gaussian with non-singular covariance matrix, SA can provide differential privacy guarantees with the level of privacy $\epsilon$ bounded by the reciprocal of the minimum eigenvalue of the covariance matrix. However, we further demonstrate that in practice, these conditions are almost unlikely to hold and hence additional noise added in model updates is still required in order for SA in FL to achieve DP. Lastly, we discuss the potential solution of leveraging inherent randomness inside aggregated model update to reduce the amount of addition noise required for DP guarantee.

%


\end{abstract}


\maketitle

\section{Introduction}
Federated learning (FL) has garnered considerable attention in recent years  due to its ability for facilitating collaborative training of machine learning models using locally private data from multiple users, eliminating the need for users to share their private local data with a central server~\cite{cc, kairouz2019advances,FedAvg}. A standard FL system involves a central server that oversees the training of a global model, which is regularly updated locally by the users over multiple rounds. In each round, the server initially distributes the current global model to the users. Subsequently, the users enhance the global model by training it on their respective private datasets and then transmit their local model updates back to the server. The server then modifies the global model by aggregating the received local model updates from the users and so on.

Although communicating the local model updates avoids sharing of data directly, it has been shown that the updates can be reverse-engineered to leak information about a user's dataset~\cite{NEURIPS2019_60a6c400,geiping2020inverting,yin2021gradients}.

To prevent information leakage from individual models, Secure Aggregation (SA) protocols have been employed to enable the server to aggregate local model updates from multiple users without having access to any clear model updates.
In each training round, users encrypt their local updates before sending them to the server for aggregation, such that the server can only learn the aggregated model, thereby preserving the privacy of individual updates. Recently, the work in \cite{elkordy2022much} extended the privacy guarantees for secure aggregation to information-theoretic guarantees on the leakage through the aggregate model over multiple training rounds. However, whether noise induced through vanilla secure aggregation (considering other participating users as noise) can provide worst-case differential privacy (DP)~\cite{dwork2014algorithmic} guarantees without any novel randomness has remained an open problem as highlighted in~\cite{elkordy2022much}.

In this paper, we target an answer to this question by focusing on formal differential privacy (DP) for secure aggregation (see Fig.~\ref{fig:fl_fig}). Specifically, we provide theoretical answers to the following: 
\begin{align*}
\parbox{3in}{\it 
\begin{enumerate}
    \item Under what conditions can secure aggregation provide DP guarantee?
    \vspace{0.2em}
    \item If it can, how much differential privacy can it offer?
    \item If not, is it possible to leverage the inherent randomness inside aggregated model update to reduce the amount of additional noise required for DP?
\end{enumerate}
}
\end{align*}

\begin{figure} 
\centering
\includegraphics[width=0.9\linewidth]{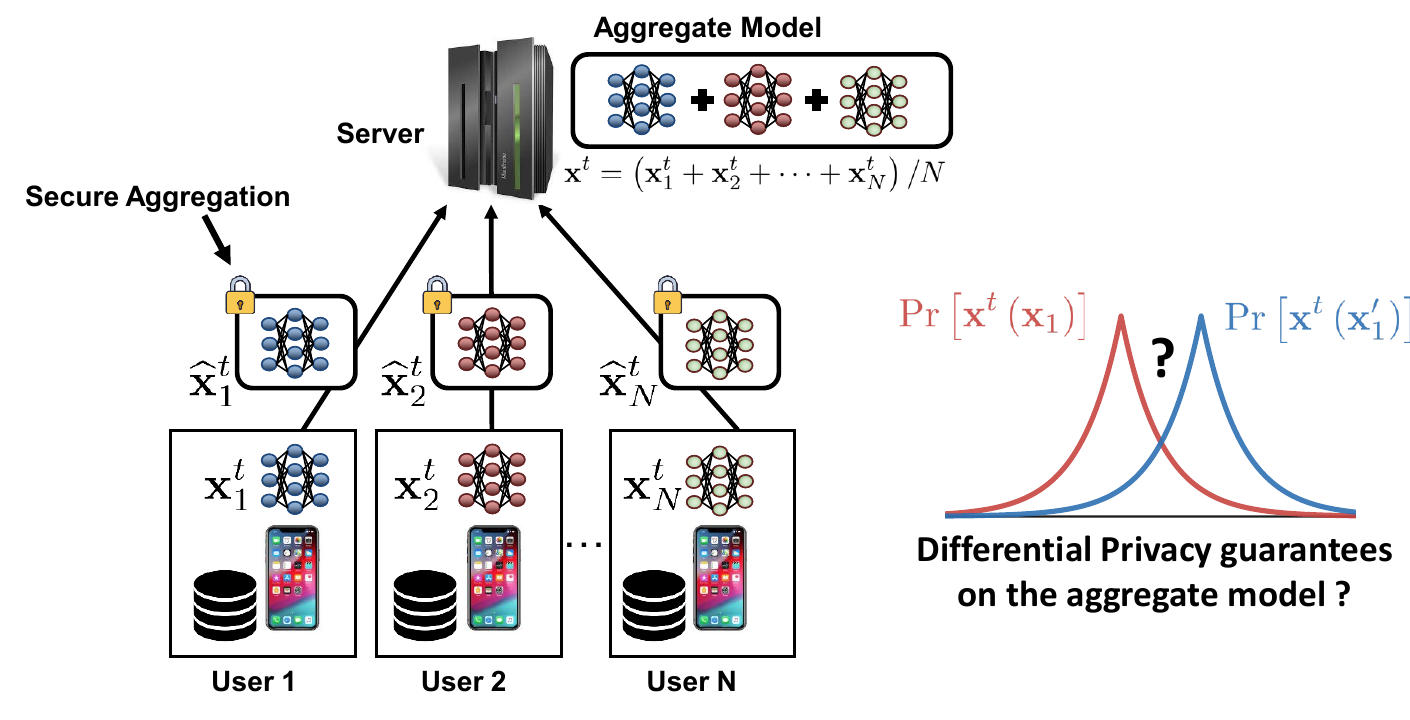}
\vspace{-.1in}
\caption{Federated learning with SA and DP guarantees.}
\vspace{-.1in}
\label{fig:fl_fig}
\end{figure}

\noindent {\bf Contributions.}
We demonstrate that one necessary condition for the aggregated model update to provide DP guarantees for individual user's local dataset is that the space of each individual user’s model update is included in the space of the aggregated model update from all users, see Sections \ref{sec:whatweneedforDP} and Section \ref{sec:theory} for detailed explanation and formal proofs. 
We further prove that under the special case where the randomness inside the aggregated model update is Gaussian with non-singular covariance matrix, the aggregated model update can provide DP guarantees for individual user's local dataset without noise additional, where the DP $\epsilon$ is bounded by the reciprocal of the minimal eigenvalue of the covariance matrix of the aggregated model update.

Moreover, we demonstrate that in practice, the conditions for SA to provide DP guarantee is almost unlikely to hold, especially in deep learning implementations where models are over-parameterized. Therefore, additional DP noise is required in order for SA to provide DP guarantee for individual user's dataset. Lastly, we investigate the possibility of leveraging the inherent randomness inside aggregated model update to reduce the amount of additional noise needed for the same level of DP.

\section{Preliminaries}
We start off by presenting the standard federated learning framework before introducing secure aggregation and its privacy guarantees, and concluding with key definitions for differential privacy which are of interest to this work.
\subsection{Federated learning}
In federated learning (see Fig.~\ref{fig:fl_fig}), a set of $N$ users collaborate with a central server to learn a machine learning model parameterized by $\theta \in \mathbb{R}^d$. Each user $i$ has its own local dataset $\mcal{D}_i$.
For simplicity, we assume that the datasets $\{\mcal{D}_i\}_{i=1}^N$ are equal in size.
The optimization problem to collaboratively solve can be formulated as:
\begin{align}
    \theta^\star = {\rm arg}\min_{\theta}\ L(\theta) = \frac{1}{N}\sum_{i=1}^N L_i(\theta).
\end{align}
To solve this problem in a distributed manner, the server and users follow an iterative protocol to train the model parameter $\theta$. In particular, at the $t$-th iteration, the server sends the latest global model $\theta^{(t)}$ and each user $i$ computes an update $x_i^{(t)}$ for that model, which is then aggregated by the server to create the new global model as follows:
\begin{align}
    \theta^{(t+1)} = \theta^{(t)} + \frac{1}{N}\sum_{i=1}^N x_i^{(t)}.
\end{align}
Two common FL protocols for computing the update are FedSGD and FedAvg \cite{mcmahan2017communication}. In the FedSGD protocol, each user samples a set of $\mcal{B}^{(t)}$ of $B$ local samples and computes the update as follows:
\begin{align}
    x_i^{(t)} = - \frac{\eta^{(t)}}{B} \sum_{b \in \mcal{B}^{(t)}} g_i(\theta^{(t)}, b),
\end{align}
where $g_i(\theta^{(t)}, b)$ is an estimate of gradient $\nabla L_i(\theta^{(t)})$, which is calculated by sampling a sample $b$ uniformly at random without replacement from $\mcal{D}_i$. In FedAvg, each user will use mini-batches in $\mcal{D}_i$ to run multiple SGD steps and then send the model update to the server.

\subsection{Secure aggregation guarantees}\label{subsection:sa_guarantee}

Secure model aggregation \cite{aono2017privacy, truex2019hybrid,dong2020eastfly,xu2019hybridalpha,secagg_bell2020secure,9712310,secagg_so2021securing,secagg_kadhe2020fastsecagg,zhao2021information,so2021lightsecagg,9712310,mugunthan2019smpai,so2021turbo} is used in conjunction with FL for providing a privacy preserving framework without sacrificing training performance.  
 SA protocols are  broadly classified into two main categories based on the encryption technique;  homomorphic encryption \cite{aono2017privacy, truex2019hybrid,dong2020eastfly,xu2019hybridalpha}, and  secure multi-party computing (MPC) \cite{secagg_bell2020secure,secagg_so2021securing,secagg_kadhe2020fastsecagg,zhao2021information,so2021lightsecagg,9712310,mugunthan2019smpai,so2021turbo}. In SA,  the encrypted model update sent by each user $i$,    $\mathbf{y}^{(t)}_i = \text{Enc}(\mathbf{x}^{(t)}_i)$ is designed to achieve two conditions: 1) Correct decoding of the aggregated model under users' dropout; 2) Privacy for the individual local model update of the users from the encrypted model.  
In the  following, we expound upon each of these conditions in a formal manner.

\noindent \textbf{Correct decoding.} The encryption mechanism ensures accurate decoding of the aggregated model of the surviving users $\mathcal{V} \subset \mathcal{N}$, even in the event of user dropout during the protocol's execution.
 In other words, the server should be able to decode 
\begin{equation}
    \text{Dec} \left(\sum_{i \in \mathcal{V}} \mathbf{y}^{(t)}_i \right)= \sum_{i \in \mathcal{V}} \mathbf{x}^{(t)}_i.
\end{equation}

\noindent \textbf{Privacy guarantee.} 
The encryption should guarantee that under the collusion of a subset of users $\mathcal{T}$     with the server,  the  encrypted model updates $\{\mathbf{y}^{(t)}_i\}
   _{i \in[N]}$ leak no information  about the model updates $\{\mathbf{x}^{(t)}_i\}
   _{i \in[N]}$ beyond the aggregated model $\sum_{i=1}^{N} \mathbf{x}^{(t)}_i$. This is formally given as follows 
\begin{equation}\label{eq-SA_guarantee}
   I\left({\{\mathbf{y}}^{(t)}_i\}
   _{i \in[N]}; \{\mathbf{x}^{(t)}_i\}
   _{i \in[N]} \middle  |   \sum_{i=1}^{N} \mathbf{x}^{(t)}_i, \mb{z}_\mathcal{T} \right) = 0,
\end{equation}
where $\mb{z}_\mcal{T}$ is the collection of information at the users in $\mcal{T}$.
The assumption of having the set of colluding users to be   $|\mathcal{T}|\leq \frac{N}{2}$  is widely used in most of the prior  works   (e.g., \cite{so2021lightsecagg,so2021turbo}). 

\subsection{Differential privacy}
Differential privacy (DP) has emerged as a reliable mechanism for ensuring data privacy in Federated Learning (FL) through the injection of noise \cite{wei2020federated}. 
When using DP, the server aggregate model still retains an added noise component, unlike the SA approach where the server recovers a non-noisy version of the aggregate. By carefully controlling the level of noise introduced through the DP mechanism, a provable privacy guarantee can be achieved for the local data of participating users, even when the local or aggregated models are accessible to adversaries. However, the required noise  for achieving high level of privacy  may negatively impact the model's performance. In the 
 following we  formally define DP.

\noindent \begin{defin}[DP \cite{dwork2014algorithmic}]: A randomized mechanism $\mathcal{M}$ :
$\mathcal{X} \rightarrow \mathrm{R}$ with domain $\mathcal{X}$ and range  $\mathrm{R}$ satisfies ($\epsilon$, $\delta$)-DP, if for
all  sets $\mathcal{S} \subseteq{R} $ and for any two adjacent databases
$\mathcal{D}_i$, $\mathcal{D}'_i$ $\in \mathcal{X}$.  
\end{defin}

\begin{equation}
\label{eq:dp_define}
    Pr [ \mathcal{M}(\mathcal{D}_i) \in \mathcal{S}]\leq e^{\epsilon} Pr [ \mathcal{M}(\mathcal{D}'_i)]+\delta.
\end{equation}
In~\eqref{eq:dp_define}, $\epsilon$ represents the level of privacy, where  smaller  $\epsilon$ means a higher level of  privacy the mechanism $\mathcal{M}$ can achieve. On the other hand, $\delta$ is a relaxation item which represents the probability of the event that the ratio
$\frac{Pr[\mathcal{M}(\mathcal{D}_i) \in \mathcal{S}]}{Pr[\mathcal{M}(\mathcal{D}'_i) \in \mathcal{S}]}$  cannot be
always bounded by $e^{\epsilon}$
after applying the privacy-preserving mechanism $\mathcal{M}$. Specifically, when $\delta=0$, we define that the mechanism $\mathcal{M}$ satisfies pure differential privacy, i.e. $\epsilon$-DP. When $\delta>0$, we define that the mechanism $\mathcal{M}$ satisfies approximate differential privacy, i.e. $(\epsilon,\delta)$-DP.
To achieve local differential privacy in FL,  Gaussian mechanism is widely adapted \cite{wei2020federated}.  In particular, in each training round, each user $i \in [N]$ clips its model update $clip(\mathbf{x}_i,C)=\frac{C\mathbf{x}_i}{max(||\mathbf{x}_i||_2,C)}$, where 
$C$ is a clipping threshold used for bounding the model update of  each user.  
Each user then adds a Gaussian noise $n$  with a  standard deviation inversely proportional to the privacy level $\epsilon$.

Another relaxation of $\epsilon$-DP is Rényi Differential Privacy (RDP), which is defined based on the Rényi divergence. We first provide the definition of Rényi divergence between two probability distributions as follows:
\begin{defin}[Rényi Divergence\cite{gil2013renyi}]
\label{rdiv}
For two probability distributions as follows: $P$ and $Q$ over a continuous space $\mathcal{X}$, and for any order $\alpha > 0$, $\alpha \neq 1$, the Rényi divergence of order $\alpha$ from $P$ to $Q$ is defined as:
\[
D_{\alpha}(P \| Q) = \frac{1}{\alpha - 1} \log \mathbb{E}_{x \sim Q}\left[ \left( \frac{P(x)}{Q(x)} \right)^\alpha \right],
\]
where $\mathbb{E}_{x \sim Q}[\cdot]$ denotes the expectation with respect to $Q$. For $\alpha = 1$, $D_{\alpha}(P \| Q)$ converges to the KL divergence as $\alpha$ approaches 1.
\end{defin}
\noindent Based on Definition \ref{rdiv}, the definition of $(\alpha, \epsilon)$-RDP is given below:
\begin{defin}[$(\alpha, \epsilon)$-RDP\cite{mironov2017renyi}]
A randomized mechanism $\mathcal{M}:\mathcal{X}\rightarrow \mathcal{Y}$ satisfies $(\alpha, \epsilon)$-RDP if for any two adjacent inputs $x, x' \in \mathcal{X}$, the following inequality holds:
\begin{equation}
    D_{\alpha}(\mathcal{M}(x) \| \mathcal{M}(x')) \leq \epsilon,
\end{equation}
where $D_{\alpha}(\cdot \| \cdot)$ denotes the Rényi divergence of order $\alpha$ between two probability distributions.
\end{defin}
Compared with $(\epsilon,\delta)$-DP, RDP provides a more accurate and convenient way to analyze and quantify privacy loss, especially when multiple mechanisms are applied \cite{mironov2017renyi}. Moreover, when RDP is satisfied, $(\epsilon,\delta)$-DP can also be guaranteed. Lemma \ref{lemma1} below provides a bridge for converting RDP guarantees into $(\epsilon,\delta)$-DP guarantees, facilitating the application of differential privacy in various settings.
\noindent \begin{lemma}[From RDP to $(\epsilon,\delta)$-DP\cite{mironov2017renyi}]
A mechanism that satisfies $(\alpha, \epsilon)$-RDP also satisfies $(\epsilon + \frac{\log(1/\delta)}{\alpha-1}, \delta)$-DP for any $\delta > 0$. 
\label{lemma1}
\end{lemma}

\subsection{Threat model for FL with SA}
\noindent\textbf{Server.} We assume that the FL server can only observe the aggregated model updates from all users without accessing individual model update from any user. We assume that the FL server is honest-but-curious. It is honest in that it will follow the FL and SA protocols without violating them. It is curious since it may be interested in inferring sensitive local dataset of each individual user from the aggregate model updates.\\
\noindent\textbf{User.} We assume that the users are honest, in that they will follow the FL with SA protocol to calculate local model updates using their local dataset and send them to the server. We assume that the users will not send the incorrect model updates to the server.

\noindent\textbf{Privacy goal.} Our privacy goal is to achieve record-level DP for each user. This says, if a user changes one record in their local dataset, the server cannot distinguish the new record from the original one, by observing the aggregated model updates across multiple training rounds.
\section{Problem Statement}
\subsection{Motivation}
\begin{figure}
    \centering
    \includegraphics[width=0.8\linewidth]{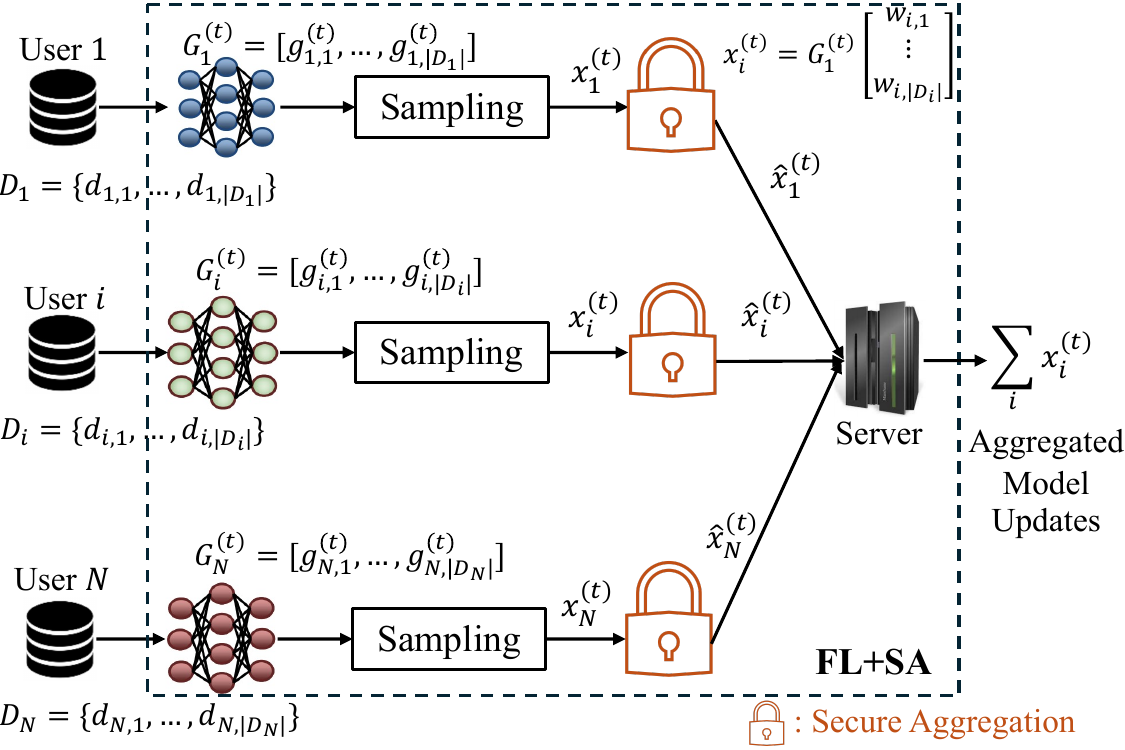}
    \caption{System model for FL with SA. Note that the input of this system is users' local datasets ($\{D_i\}_{i=1}^{i=N}$), and the output of the system is the aggregated model update ($\sum_{i=1}^{i=N}x_i^{(t)}$), which is a random vector due to users' local gradient (i.e. data batch) sampling. The server will infer user $i$'s local dataset ($D_i$) by observing $\sum_{i=1}^{i=N}x_i^{(t)}$.}
    \label{fig:sys_model}
\end{figure}
We define the local dataset of user $i$ as $D_i$, the model update from user $i$ at step $t$ as $x_{i}^{(t)}$, and the aggregated model update from all $N$ users excluding user $i$ as $x_{-i}^{(t)}=\sum_{j=1,j\neq i}^{j=N} x_{j}^{(t)}$. As demonstrated in Figure \ref{fig:sys_model}, when SA is used in FL, the server only receives  $x_{i}^{(t)}+x_{-i}^{(t)}$ from the aggregator (i.e. the output of the FL+SA system) without knowing $x_{i}^{(t)}$ from user $i$. Since each user conducts the data batch sampling when calculating local model update, the aggregated model update $x_{i}^{(t)}+x_{-i}^{(t)}$ can be viewed as a random vector, which can provide some privacy protection to user $i$'s local dataset $D_i$. Moreover, compared with the case where SA is not used and server knows $x_{i}^{(t)}$, user $i$ may add less or even no additional noise locally by oneself to achieve the same level of privacy, due to the leverage of noise from $x_{-i}^{(t)}$.

Prior work in \cite{elkordy2022much} uses Mutual Information (MI) \cite{ash2012information} as a metric to measure the privacy leakage in federated learning with secure aggregation. The authors have demonstrated that the model update from each user contains random noise, which can offer privacy protection for other users. Hence, as we aggregate more users, the MI privacy leakage of the individual model update from the aggregated model update will decrease, since more random noise is aggregated. 

Formally, the authors in \cite{elkordy2022much} use the MI between $D_{i}$ and $x_{i}^{(t)}+x_{-i}^{(t)}=\sum_{j=1}^{j=N} x_{j}^{(t)}$ as the privacy metric, which is defined as follows:

\begin{equation}
I_{\rm priv}^{(t)} = \max_{i\in[N]} I\left(D_{i} ;  \sum_{i=1}^Nx_i^{(t)}\middle|  \left\{\sum_{i=1}^Nx_i^{(k)}\right\}_{k\in[t-1]} \right). 
\end{equation}
For users with identical and independent data distributions, under the assumption that the model update can be composed of independent random variables after whitening, they prove that 
\begin{equation}\label{eq:I_priv_round}
I_{\rm priv}^{(t)} \leq \frac{C\ d}{(N-1)B} + \frac{d}{2}\log\left(\frac{N}{N-1}\right), 
\end{equation}
where $C$ is a constant, $d$ is the model size, and $B$ is the mini-batch size at training round $t$ in FedSGD. The inequality~\eqref{eq:I_priv_round} indicates that MI between the individual model update and the aggregated model update (i.e. the privacy leakage) reduces linearly with the number of users participating in FL.

However, MI only measures the on-average privacy leakage, without providing any privacy guarantees for the worst-case scenarios. Since worst-case privacy guarantees are stronger than on-average privacy guarantees because they guarantee privacy protection even in situations that are less likely to occur, it is important to investigate how much privacy can leak in the worse-case when combining SA with FL.


Formally, we define the probabilistic distribution of aggregated model update $x_{i}^{(t)}+ x_{-i}^{(t)}$ at step $t$ as $\mathcal{M}_{\{D_i\cup D_{-i},\theta^{(t)}\}}$, where:
\begin{equation}
\label{eq:dist}
    \mathcal{M}_{\{D_i\cup D_{-i},\theta^{(t)}\}}(x) = \mathtt{Pr}[x_{i}^{(t)}+ x_{-i}^{(t)}=x|\{D_i\}_{i=1}^{i=N},\theta^{(t)}],
\end{equation}
where $D_{-i}=\bigcup_{\substack{j=1,j \neq i}}^{N} D_i$. Suppose $D_i$ and $D_i^{'}$ are two instances of user $i$'s local dataset. If there exists $\alpha>1$ and $\epsilon>0$, such that for any $D_i$ and $D_i^{'}$ differing from one data point, the following inequality holds: 
\begin{equation}
    D_{\alpha}(\mathcal{M}_{\{D_i\cup D_{-i},\theta^{(t)}\}}||\mathcal{M}_{\{D_i^{'}\cup D_{-i},\theta^{(t)}\}})\leq \epsilon,
\end{equation}
then $x_{i}^{(t)}+ x_{-i}^{(t)}$ can provide $(\alpha,\epsilon)$-RDP guarantee and hence $(\epsilon + \frac{\log(1/\delta)}{\alpha-1}, \delta)$-DP for user $i$'s dataset $D_i$.

\subsection{Negative result for DP}
We start by using the counterexample from \cite{elkordy2022much} to show that the aggregated model update can not offer DP guarantee. Specifically, Figure~\ref{fig:counter_example} (from \cite{elkordy2022much}) demonstrates a worst-case scenario where the last 1/4 elements of model updates for user 1, 2, 3 (i.e. $x_1, x_2$ and $x_3$) are all zeros while the last 1/4 elements of model update from user 4 (i.e. $x_4$) are non-zero. Under this case, suppose user 4 changes one data point in their local dataset, making $x_4$ become $x_4'$. Then, an adversary can perfectly distinguish $\sum_{i=1}^{i=3}x_i+x_4$ from $\sum_{i=1}^{i=3}x_i+x_4'$. Hence, the aggregated model update can not provide DP guarantee for user 4's local dataset. 

While the privacy leakage in the worst-case scenarios can be significant, as long as the occurrence probability of worst-case scenarios is small, a reasonable on-average privacy guarantee (e.g. measured by MI) can be achieved. This is the key difference between worst-case privacy guarantee and on-average privacy guarantee.

\begin{figure}
\centering
\includegraphics[width=0.4\textwidth]{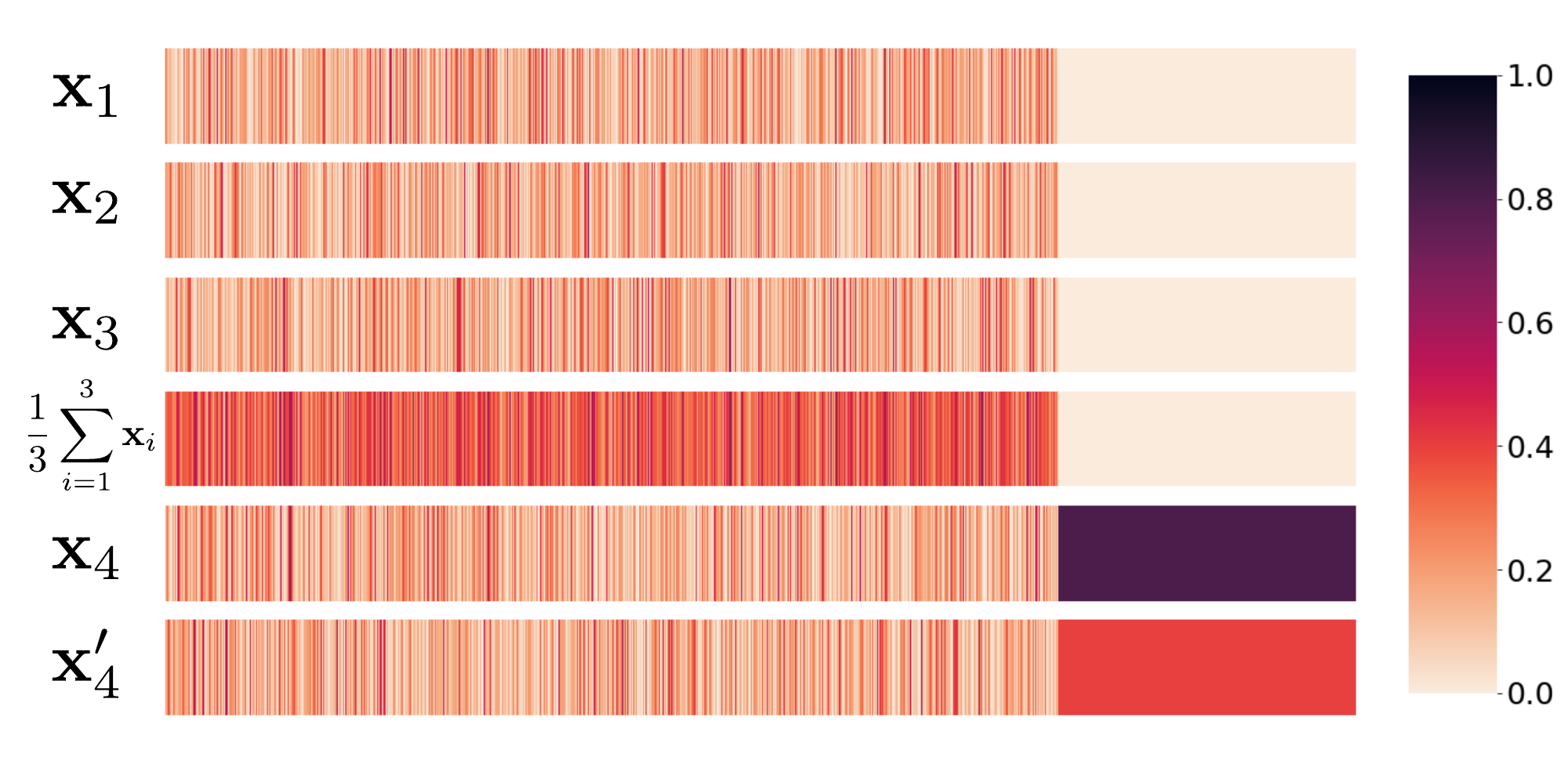}
\vspace{-.1in}
\caption{Heatmap of the absolute values of sampled updates from users $1,2$ and $3$ in the counterexample. $x_4$ and $x_4'$ can be distinguished even adding the aggregated noise from $\sum_{i=1}^3 x_i$. }
\vspace{-.1in}
\label{fig:counter_example}
\end{figure}

\subsection{What we need for DP guarantee}
\label{sec:whatweneedforDP}
Based on the above counterexample, a natural question is under what conditions the above worst case can be avoided, such that DP can be guarantee. In general, we need the following two conditions to be held:
\begin{enumerate}
    \item The random noise among the aggregated model update from others users should be independent of the model update of each individual user.
    \item The space of each individual user's model update should be included in the space of the aggregated model update from all the other users.
\end{enumerate}
The first condition is necessary since it stems from the fundamental ``context-unaware" requirements in DP \cite{dwork2014algorithmic}. The second condition is necessary for having a finite $\epsilon$. Specifically, as demonstrated in inequality~\eqref{eq:dp_define}, for any $x_{i}^{(t)}$ and $x_{i}^{(t)'}$ calculated from $D_i$ and $D_i^{'}$ differing from one data point, both $Pr[x_{i}^{(t)}+x_{-i}^{(t)}\in X] $ and $Pr[x_{i}^{(t)'}+x_{-i}^{(t)}\in X] $ should be non-zero. Otherwise, $\epsilon$ can be infinite. This says, the space of $x_{i}^{(t)}$ (i.e. the set consisting of all possible $x_{i}^{(t)}$) must be included in the space of $x_{-i}^{(t)}$ (i.e. the set consisting of all possible $x_{-i}^{(t)}$).

Motivated by the above statements on general conditions for achieving DP, in the next section we formalize several specific conditions under which we can analytically derive the $\epsilon$ bound of DP in FL with SA.


\section{Theoretical Results}
\label{sec:theory}
In this section, we analyze under what conditions FL with SA can provide DP guarantee when using FedSGD protocol. 

\subsection{Basic assumption for gradient noise}
We start by making the following assumption about FedSGD protocol in our analysis, which has been widely used in FL \cite{kairouz2019advances}.

\begin{assump}[Stochastic gradient with sampling noise and clipping]
Assume that each user $i$ at step $t$ uses its local dataset $D_{i}$ with size $|D_{i}|$ to compute the model update as follows:
\begin{equation}
\label{eq:fedsgd1}
x_{i}^{(t)} = G_{i}^{(t)}w_i=\begin{bmatrix} g_{i,1}^{(t)} & \cdots & g_{i,|D_{i}|}^{(t)} \end{bmatrix} \begin{bmatrix} w_{i,1} \\ \vdots \\ w_{i,|D_{i}|} \end{bmatrix},
\end{equation}
where 
$g_{i,j}^{(t)} = clip(\frac{\partial L(\theta^{(t)},d_{i,j})}{\partial \theta^{(t)}}, C)\in\mathbf{R}^{K}$ is clipped gradient computed using $d_{i,j}$, $K$ is the dimension of gradient (i.e. model size), $d_{i,j}$ is the $j$-th data point in user $i$'s local dataset $D_{i}$, $clip(x,C)=\frac{Cx}{max(||x||_2,C)}$ is a clipping function to bound the norm of vector $x$ by constant $C$, and $w_i\in\mathbf{R}^{|D_{i}|}$ is a random vector representing the sampling noise, which satisfies $\mathbb{E}_{w_i}[x_{i}^{(t)}]=\frac{\sum_{j=1}^{|D_{i}|}g_{i,j}^{(t)}}{|D_{i}|}$. Without loss of generality, we assume $|D_{i}|=D$ for all users.
\label{assum1}
\end{assump}
\noindent\textbf{Remark 1: Applicability of Assumption \ref{assum1}.} Assumption \ref{assum1} assumes that the gradient randomness is due to sampling ``noise", which holds in practice. For example, in traditional SGD, user $i$ will randomly and uniformly sample a mini-batch with size $B$ from its local dataset to calculate the model update at each step (denoted as IID sampling). In this case, each element in $w_i$ will be a random variable satisfying $w_{i,j} \in \{0, \frac{1}{B}\}$ and $\sum_{j=1}^{D} w_{i,j} = 1$, This says, for user $i$, we have $w_{i,j} = \frac{1}{B}$ for exactly $B$ indices $j$, and $w_{i,j} = 0$ for the other $D-B$ indices.
As another example, when user $i$ uses the Gaussian sampling noise proposed in \cite{wu2020noisy}, $w_i\sim N(\frac{1}{D}, \frac{1}{BD}I_D)$ is a Gaussian random vector. Depending on the sampling strategy, the distribution of the gradient noise in SGD will be different. 

The other assumption in Assumption \ref{assum1} is that users will clip their gradients, which is widely used in privacy-preserving machine learning, see, for example, \cite{abadi2016deep}.

Based on Assumption \ref{assum1}, from the perspective of user $i$, given its local dataset 
$D_{i}$ as input of the FL with SA system, the output of the system will be the aggregated model update calculated as:
\begin{equation}
\label{eq:global}
\begin{aligned}
    x^{(t)}=
    &x_{i}^{(t)}+ x_{-i}^{(t)}
    =G_{i}^{(t)}w_i + \sum_{j=1,j\neq i}^{N}G_{j}^{(t)}w_j=\sum_{j=1}^{N}G_{j}^{(t)}w_j \\
\end{aligned}
\end{equation}
where $w_1,...,w_N$ are independent with each other. 

%
\subsection{Necessary condition for DP guarantee}
We assume that $x_{i}^{(t)}$ and $x_{i}^{(t)'}$ are two model update instances of $x_{i}^{(t)}$ calculated from $D_i$ and $D_i^{'}$ differing from one data point respectively. Suppose that $D_i$ and $D_i^{'}$ differs from the $j$-th data point, we define the $j$-th data point in $D_i$ and $D_i^{'}$ as $d_{i,j}$ and $d_{i,j}^{'}$ respectively, and the corresponding gradient calculated from the $j$-th data point as $g_{i,j}^{(t)}$ and $g_{i,j}^{(t)'}$ respectively. Moreover, we define the span of vectors $\{v_1,..,v_K\}$ as $span(\{v_1,..,v_K\})=\{x_1v_1 + ... + x_Kv_K| \forall~x_1,...,x_K\in\mathbf{R}\}$ .

As motivated in Section \ref{sec:whatweneedforDP}, in order for the aggregated model update to guarantee $(\epsilon,\delta)$-DP with bounded $\epsilon$ value, for any $x_{i}^{(t)}$ and $x_{i}^{(t)'}$, and any $X\in Range(x_{i}^{(t)}+x_{-i}^{(t)})$ satisfying $Pr[x_{i}^{(t),1}+x_{-i}^{(t)}\in X]\geq \delta$, both $Pr[x_{i}^{(t),1}+x_{-i}^{(t)}\in X] $ and $Pr[x_{i}^{(t),2}+x_{-i}^{(t)}\in X] $ should be non-zero. Motivated by the above, we derive the following theorem.

\begin{thm}[A necessary condition for DP guarantee]
Under Assumption \ref{assum1}, if the aggregated model update $x_{i}^{(t)}+x_{-i}^{(t)}$ can provide $(\epsilon,\delta)$-DP user $i$'s dataset $D_i$ and $\delta \leq \mathtt{Pr}[w_{i,j}\neq 0]$, then for any $D_i$ and $D_i^{'}$ differs from the $j$-th data point, we must have $g_{i,j}^{(t)'}\in span\big(\bigcup_{j=1}^{N} \{g_{j,k}^{(t)}|\forall~k\in\{1,...,D\}\}\big)$.
\label{theorem_base}
\end{thm}
\begin{proof}
Assume that there exists a data point $d_{i,j}^{'}$, such that $g_{i,j}^{(t)'}\notin span\big(\bigcup_{j=1}^{N} \{g_{j,k}^{(t)}|\forall~k\in\{1,...,D\}\}\big)$ when user $i$ changes $d_{i,j}$ into $d_{i,j}^{'}$. The aggregated model update when user $i$ has local dataset $D_i$ is calculated as:
\begin{equation}
\label{eq:aggre_update1}
\begin{aligned}
    x_{i}^{(t)}+ x_{-i}^{(t)} &=g_{i,j}^{(t)}w_{i,j} + \sum_{k=1,k\neq j}^{D}g_{i,k}^{(t)}w_{i,k} + \sum_{j=1,j\neq i}^{N}\sum_{k=1}^{D}g_{j,k}^{(t)}w_{j,k}
\end{aligned}
\end{equation}
and the aggregated model update when user $i$ has local dataset $D_i^{'}$ is calculated as:
\begin{equation}
\label{eq:aggre_update2}
\begin{aligned}
    x_{i}^{(t)'}+ x_{-i}^{(t)} &=g_{i,j}^{(t)'}w_{i,j} + \sum_{k=1,k\neq j}^{D}g_{i,k}^{(t)}w_{i,k} + \sum_{j=1,j\neq i}^{N}\sum_{k=1}^{D}g_{j,k}^{(t)}w_{j,k}
\end{aligned}
\end{equation}
Next, we define $X$ as $\{x_{i}^{(t)'}+ x_{-i}^{(t)}|\forall~w_{i,j}\neq 0\}$. We can derive that $\mathtt{Pr}[x_{i}^{(t)'}+ x_{-i}^{(t)}\in X]=\mathtt{Pr}[w_{i,j}\neq 0]$.
Since 
$$g_{i,j}^{(t)'}\notin span\big(\bigcup_{j=1}^{N} \{g_{j,k}^{(t)}|\forall~k\in\{1,...,D\}\}\big),$$ 
and 
$$x_{i}^{(t)}+ x_{-i}^{(t)}\in span\big(\bigcup_{j=1}^{N} \{g_{j,k}^{(t)}|\forall~k\in\{1,...,D\}\}\big),$$ 
then $\forall x\in X$, we have $\mathtt{Pr}[x_{i}^{(t)}+ x_{-i}^{(t)}=x]=0$. This says $\mathtt{Pr}[x_{i}^{(t)}+ x_{-i}^{(t)}\in X]=0$.
Then, $\forall~\epsilon\in\mathbf{R}$, we have $\mathtt{Pr}[x_{i}^{(t)'}+ x_{-i}^{(t)}\in X]-e^\epsilon\mathtt{Pr}[x_{i}^{(t)}+ x_{-i}^{(t)}\in X]=\mathtt{Pr}[w_{i,j}\neq 0]\geq \delta$. Therefore, $(\epsilon,\delta)$-DP cannot be guaranteed for $\delta \leq \mathtt{Pr}[w_{i,j}\neq 0]$ \footnote{Note that $\mathtt{Pr}[w_{i,j}\neq 0]$ indicates the probability that the $j$-th data point is sampled to calculate the model update from user $i$ (e.g. $\mathtt{Pr}[w_{i,j}\neq 0]=\frac{B}{D}$ in IID sampling and $\mathtt{Pr}[w_{i,j}\neq 0]=1$ in Gaussian sampling). Although $(\epsilon,\delta)$-DP may hold for $\delta > \mathtt{Pr}[w_{i,j}\neq 0]$, $\mathtt{Pr}[w_{i,j}\neq 0]$ is significantly larger than a meaningful bound for $\delta$ \cite{dwork2014algorithmic}.}.
\end{proof}

Theorem \ref{theorem_base} states that a necessary condition for aggregated model update to provide DP guarantee for each user's local dataset is that any change of one data point in user $i$'s local dataset will not change the span of all individual gradients from all users. Since the aggregated model update belongs to the span of all individual gradients (see Eq. (\ref{eq:aggre_update1})), when the server observes that the aggregated model is from a different span, it can potentially identify or reconstruct the data point which causes this change. Therefore, the worst-case privacy DP guarantee is violated. Recent works \cite{pasquini2022eluding} have proposed attacks to elude SA in FL via model inconsistency across clients or model manipulation. Fundamentally, these attacks work by making the model updates from different users fall into different vector space and hence the necessary condition for worst-privacy guarantee in Theorem \ref{theorem_base} will be violated, causing privacy leakage.

\subsection{Gaussian sampling noise with non-singular covariance matrix}
\label{subsec:gaussian_non_singular}
The most common type of sampling noise is IID sampling. Specifically, consider the scenario where each user $i$ will randomly and uniformly sample a mini-batch with size $B$ from its local dataset to calculate the model update at each step. We define $span^{\frac{1}{B}}\big(\{v_1,..,v_K\}\big)=\{x_1v_1 + ... + x_Kv_K| \forall~x_1,...,x_K\in \{0, \frac{1}{B}\}~and~\sum_{k=1}^{K}x_k=1\}$. Since in practice, each user has a limited amount of data points in their local dataset, the aggregated model update $x_{i}^{(t)}+ x_{-i}^{(t)}$ will be a random vector which takes value from a finite set (i.e. $span^{\frac{1}{B}}\big(\bigcup_{j=1}^{N} \{g_{j,k}^{(t)}|\forall~k\in\{1,...,D\}\}\big)$). Without assuming the distribution of user $i$'s dataset, $g_{i,j}^{(t)'}$ can have infinite amount of possible values. In this case, $g_{i,j}^{(t)'}\in span^{\frac{1}{B}}\big(\bigcup_{j=1}^{N} \{g_{j,k}^{(t)}|\forall~k\in\{1,...,D\}\}\big)$ cannot be guaranteed and hence DP is violated based on Theorem \ref{theorem_base}.

In the rest of this section, we consider a special type of sampling noise called Gaussian sampling noise, where the weights of individual gradients in model update of each user $i$ are Gaussian. We prove that a closed-form $\epsilon$ bound in DP can be derived for Gaussian sampling noise when the covariance matrix of sampling noise of each user (or the projected covariance matrix) is non-singular. 


%
\begin{assump}[Gaussian sampling noise]
Based on Assumption \ref{assum1}, we assume that each user $i$ uses Gaussian sampling noise, which satisfies $w_i=\frac{1}{D} + \frac{1}{\sqrt{BD}}\xi_i$ and $\xi_i\sim N(0, I_D)$. Then, the covariance matrix of gradient noise is derived as:
\begin{equation}
\begin{aligned}
 \Sigma_i^{(t)}&=\frac{1}{D}G_{i}^{(t)}(G_{i}^{(t)})^T=\frac{1}{D}\sum_{j=1}^{j=D}g_{i,j}^{(t)}(g_{i,j}^{(t)})^T
\end{aligned}
\end{equation}
We define the SVD decomposition of covariance matrix as $\Sigma_i^{(t)}=U_{i}^{(t)}\lambda_{i}^{(t)}\big(U_{i}^{(t)}\big)^T$. 
Then, the model update from user $i$ at step $t$ can be computed as 
\begin{equation}
\label{eq:gaussin_gradient}
    x_{i}^{(t)} = \bar g_{i}^{(t)} + \frac{1}{\sqrt{BD}}G_{i}^{(t)}\xi_i = \bar g_{i}^{(t)} + \frac{1}{\sqrt{B}}L_{i}^{(t)}v_i,
\end{equation}
where $\bar g_{i}^{(t)}=\frac{1}{D}\sum_{j=1}^{D}g_{i,j}^{(t)}$, $L_{i}^{(t)}=U_{i}^{(t)}\big(\lambda_{i}^{(t)}\big)^{\frac{1}{2}}$, 
$v_i\sim \mathcal{N}(0, I_K)$ ($K$ is the dimension of model update), $v_1,...,v_N$ are independent with each other, and $B$ is a scaling factor to adjust the magnitude of the covariance\footnote{Note that when $B$ is set as the mini-batch size in classical SGD (i.e. take the average of $B$ IID sampled gradient as model update), the covariance of Gaussian gradient noise will be close to the covariance of classical SGD noise \cite{wu2020noisy}.}.
\label{assum2}
\end{assump}
Based on Assumption \ref{assum2}, the aggregated model update at step $t$ can be written as:
\begin{equation}
\label{eq:global}
\begin{aligned}
    x^{(t)}=
    &\sum_{j=1}^{N}G_{j}^{(t)}w_j
    =\sum_{j=1}^{N}\bar g_{j}^{(t)}+ \sum_{j=1}^{N}\frac{1}{\sqrt{B}}L_{i}^{(t)}v_j. \\
\end{aligned}
\end{equation}
Moreover, given that $v_1,...,v_N$ are independent normal Gaussian, $x^{(t)}$ will also be Gaussian with mean $\sum_{j=1}^{N}\bar g_{j}^{(t)}$ and covariance matrix $\frac{\sum_{j=1}^{N}\Sigma_j^{(t)}}{B}$.
\noindent \begin{lemma}[Bounded maximal singular value] 
\label{lemma2}
    Based on Assumption \ref{assum1}, for each user $i$, the maximal singular value of the gradient covariance matrix of gradient noise $\Sigma_i^{(t)}$ is upper bounded by $C^2$.
\end{lemma}
\begin{proof}
    $\forall~x\in\mathbf{R}^K$, we have $x^T\Sigma_i^{(t)}x=\frac{1}{D}\sum_{j=1}^{j=D}x^T(g_{i,j}^{(t)})^Tg_{i,j}^{(t)}x=\frac{1}{D}\sum_{j=1}^{j=D}(g_{i,j}^{(t)}x)^2$. Since $|g_{i,j}^{(t)}x|$$\leq||g_{i,j}^{(t)}||_2||x||_2\leq C||x||_2$, we have $x^T\Sigma_i^{(t)}x\leq C^2||x||_2^2$. Hence, the maximal singular value of $\Sigma_i^{(t)}$ is upper bounded by $C^2$.
\end{proof}
\begin{assump}[Non-singular covariance matrix]
We assume that for each user $i$, the covariance matrix of gradient noise  $\Sigma_i^{(t)}$ is non-singular, with non-zero minimal eigenvalue lower bounded by $\lambda_{i,min}^{(t)}>0$.
\label{assum3}
\end{assump}
\begin{thm}
Under Assumption \ref{assum1},\ref{assum2},\ref{assum3}, if $\frac{C^2}{D} < \sum_{i=1}^{N}\lambda_{i,min}^{(t)}$, then $\forall~\alpha\in (1, \frac{D\sum_{i=1}^{N}\lambda_{i,min}^{(t)}}{C^2})$, the aggregated model update $x^{(t)}$ at step $t$ can provide $(\alpha,\epsilon_i^{(t)})$-RDP and $(\epsilon_i^{(t)}+\frac{\log(1/\delta)}{\alpha-1},\delta)$-DP for user $i$'s dataset $D_i$, where:
\begin{equation}
\label{eq:lambda1}
\epsilon_i^{(t)} =
    (\frac{2\alpha BC^2}{D^2}+\frac{\alpha C^2}{(\alpha-1)D})\frac{1}{\sum_{i=1}^{N}\lambda_{i,min}^{(t)}-\frac{\alpha C^2}{D}}.
\end{equation}
\label{theorem1}
\end{thm}
\noindent\begin{proof}
Under Assumption \ref{assum1} and \ref{assum2}, based on Eq. (\ref{eq:dist}) and Eq. (\ref{eq:global}), for any $D_i$ and $D_i^{'}$ differing from one data point,
the Rényi divergence between $\mathcal{M}_{\{D_i\cup D_{-i},\theta^{(t)}\}}$ (i.e. the probabilistic distribution of aggregated model update when $D_i$ is used) and $\mathcal{M}_{\{D_i^{'}\cup D_{-i},\theta^{(t)}\}}$ (i.e. the probabilistic distribution of aggregated model update when $D_i^{'}$ is used) is derived as the Rényi divergence between two Gaussian distribution:
\begin{equation}
\label{eq:rd_bound}
\begin{aligned}
&D_{\alpha}(\mathcal{M}_{\{D_i\cup D_{-i},\theta^{(t)}\}} \| \mathcal{M}_{\{D_i^{'}\cup D_{-i},\theta^{(t)}\}}) \\
= &D_{\alpha}(\mathcal{N}(\bar g_{i}^{(t)}+\bar g_{-i}^{(t)}, \frac{\Sigma_i^{(t)}+\Sigma_{-i}^{(t)}}{B}) \|\mathcal{N}(\bar g_{i}^{(t)'}+\bar g_{-i}^{(t)}, \frac{\Sigma_i^{(t)'}+\Sigma_{-i}^{(t)}}{B}))
\end{aligned}
\end{equation}

Without loss of generality, we assume that $D_i$ and $D_i^{'}$ differs from the first data point, which are denoted as $d_{i,1}$ and $d_{i,1}^{'}$ respectively, and the corresponding gradients calcualted by these two data points are $g_{i,1}^{(t)}$ and $g_{i,1}^{(t)'}$. We define $\Delta \bar g_{i}^{(t)} = \bar g_{i}^{(t)}-\bar g_{i}^{(t)'} = \frac{1}{D}(g_{i,1}^{(t)}-g_{i,1}^{(t)'})$, $\Sigma_1=\frac{\Sigma_i^{(t)}+\Sigma_{-i}^{(t)}}{B}$, $\Sigma_2=\frac{\Sigma_i^{(t)'}+\Sigma_{-i}^{(t)}}{B}$, and $\Sigma_\alpha=(1-\alpha)\Sigma_1 + \alpha\Sigma_2$.

Since $\forall~x\in\mathbf{R}^K$,
\begin{equation}
\begin{aligned}
&x^T\Sigma_\alpha x=x^T((1-\alpha)\Sigma_1 + \alpha\Sigma_2)x \\
=&x^T\Sigma_1 x + x^T\alpha(\Sigma_2-\Sigma_1)x \\
=&x^T\Sigma_1 x + x^T\frac{\alpha}{B}(\Sigma_i^{(t)'}-\Sigma_i^{(t)})x\\
=&x^T\Sigma_1 x + x^T\frac{\alpha}{BD}(g_{i,1}^{(t)'}(g_{i,1}^{(t)'})^T-g_{i,1}^{(t)}(g_{i,1}^{(t)})^T)x \\
\geq & (\frac{\sum_{i=1}^{N}\lambda_{i,min}^{(t)}}{B}-\frac{\alpha C^2}{BD})x^Tx,
\end{aligned}
\end{equation}
and $\frac{\alpha C^2}{D} < \sum_{i=1}^{N}\lambda_{i,min}^{(t)}$,
we know that $\Sigma_\alpha$ is positive definite and its minimal eigenvalue is lower bounded by $\frac{\sum_{i=1}^{N}\lambda_{i,min}^{(t)}}{B}-\frac{\alpha C^2}{BD}$.
Therefore, based on Table 2 in \cite{gil2013renyi}, we have:
\begin{equation}
\label{eq:rd_bound2}
\begin{aligned}
&D_{\alpha}(\mathcal{N}(\bar g_{i}^{(t)}+\bar g_{-i}^{(t)}, \frac{\Sigma_i^{(t)}+\Sigma_{-i}^{(t)}}{B}) \|\mathcal{N}(\bar g_{i}^{(t)'}+\bar g_{-i}^{(t)}, \frac{\Sigma_i^{(t)'}+\Sigma_{-i}^{(t)}}{B})) \\
= & \underbrace{\frac{\alpha}{2}(\Delta \bar g_{i}^{(t)})^T \Sigma_\alpha^{-1}\Delta \bar g_{i}^{(t)}}_{Term~I}\underbrace{ -\frac{1}{2(\alpha - 1)}\ln\frac{|\Sigma_\alpha|}{|\Sigma_1|^{1-\alpha}|\Sigma_2|^{\alpha}}}_{Term~II}.
\end{aligned}
\end{equation}
First, we calculate the upper bound for Term I. Since $||\Delta \bar g_{i}^{(t)}||_2^2=||\frac{1}{D}(g_{i,1}^{(t)}-g_{i,1}^{(t)'})||_2^2 \leq \frac{4C^2}{D^2}$ and $\Sigma_\alpha$ has minimal eigenvalue lower bounded by $\frac{\sum_{i=1}^{N}\lambda_{i,min}^{(t)}}{B}-\frac{\alpha C^2}{BD}$,
we have:
\begin{equation}
\label{eq:rd_bound1_1}
\begin{aligned}
\frac{\alpha}{2}(\Delta \bar g_{i}^{(t)})^T \Sigma_\alpha^{-1}\Delta \bar g_{i}^{(t)} \leq \frac{2\alpha C^2}{D^2(\frac{\sum_{i=1}^{N}\lambda_{i,min}^{(t)}}{B}-\frac{\alpha C^2}{BD})}.
\end{aligned}
\end{equation}
Next, we calculate the upper bound for Term II. Due to the concavity of $\ln |X|$ on positive semi-definite matrix, we have
\begin{equation}
\label{eq:rd_bound2}
\begin{aligned}
&\ln\frac{|\Sigma_1|^{1-\alpha}|\Sigma_2|^{\alpha}}{|\Sigma_\alpha|}=(1-\alpha)\ln|\Sigma_1| + \alpha\ln|\Sigma_2| - \ln|\Sigma_\alpha| \\
= &\alpha(\ln|\Sigma_2| - \ln|\Sigma_1|) + (\ln|\Sigma_1|-\ln|\Sigma_\alpha|)\\
\leq & \alpha\mathbf{tr}(\Sigma_1^{-1}(\Sigma_2 - \Sigma_1)) + \mathbf{tr}(\Sigma_\alpha^{-1}(\Sigma_1 - \Sigma_\alpha)) \\
= & \alpha\mathbf{tr}(\Sigma_1^{-1}(\frac{1}{BD}g_{i,1}^{(t)'}(g_{i,1}^{(t)'})^T - \frac{1}{BD}g_{i,1}^{(t)}(g_{i,1}^{(t)})^T)) \\
& +\mathbf{tr}(\Sigma_\alpha^{-1}(\frac{\alpha}{BD}g_{i,1}^{(t)}(g_{i,1}^{(t)})^T - \frac{\alpha}{BD}g_{i,1}^{(t)'}(g_{i,1}^{(t)'})^T))) \\
= & \frac{\alpha}{BD}\mathbf{tr}(\Sigma_1^{-1}(g_{i,1}^{(t)'}(g_{i,1}^{(t)'})^T)) - \frac{\alpha}{BD}\mathbf{tr}(\Sigma_1^{-1}(g_{i,1}^{(t)}(g_{i,1}^{(t)})^T)) \\
& +\frac{\alpha}{BD}\mathbf{tr}(\Sigma_\alpha^{-1}(g_{i,1}^{(t)}(g_{i,1}^{(t)})^T)) - \frac{\alpha}{BD}\mathbf{tr}(\Sigma_\alpha^{-1}(g_{i,1}^{(t)'}(g_{i,1}^{(t)'})^T)) \\
\end{aligned}
\end{equation}
where $\mathbf{tr}(X)$ denotes the trace of matrix $X$.
Given that both $g_{i,1}^{(t)}(g_{i,1}^{(t)})^T$ and $g_{i,1}^{(t)'}(g_{i,1}^{(t)'})^T$ are semi-positive definite matrix, based on von Neumann's trace inequality \cite{mirsky1975trace}, we have:
\begin{equation}
    \mathbf{tr}(\Sigma_1^{-1}(g_{i,1}^{(t)'}(g_{i,1}^{(t)'})^T))\leq \frac{\mathbf{tr}(g_{i,1}^{(t)'}(g_{i,1}^{(t)'})^T)}{\displaystyle\min_{j=1,...,K}\lambda_{1,j}}\leq \frac{C^2}{\displaystyle\min_{j=1,...,K}\lambda_{1,j}},
\end{equation}
\begin{equation}
    \mathbf{tr}(\Sigma_1^{-1}(g_{i,1}^{(t)}(g_{i,1}^{(t)})^T))\geq \frac{\mathbf{tr}(g_{i,1}^{(t)}(g_{i,1}^{(t)})^T)}{\displaystyle\max_{j=1,...,K}\lambda_{1,j}}\geq 0,
\end{equation}
\begin{equation}
    \mathbf{tr}(\Sigma_\alpha^{-1}(g_{i,1}^{(t)}(g_{i,1}^{(t)})^T))\leq \frac{\mathbf{tr}(g_{i,1}^{(t)}(g_{i,1}^{(t)})^T)}{\displaystyle\min_{j=1,...,K}\lambda_{\alpha,j}}\leq \frac{C^2}{\displaystyle\min_{j=1,...,K}\lambda_{\alpha,j}},
\end{equation}
\begin{equation}
    \mathbf{tr}(\Sigma_\alpha^{-1}(g_{i,1}^{(t)'}(g_{i,1}^{(t)'})^T))\geq \frac{\mathbf{tr}(g_{i,1}^{(t)'}(g_{i,1}^{(t)'})^T)}{\displaystyle\max_{j=1,...,K}\lambda_{\alpha,j}}\geq 0,
\end{equation}
Therefore, we can derive:
\begin{equation}
\label{eq:rd_bound2_1}
\begin{aligned}
&\ln\frac{|\Sigma_1|^{1-\alpha}|\Sigma_2|^{\alpha}}{|\Sigma_\alpha|} \leq\frac{\alpha C^2}{BD}(\frac{1}{\displaystyle\min_{j=1,...,K}\lambda_{1,j}} + \frac{1}{\displaystyle\min_{j=1,...,K}\lambda_{\alpha,j}})\\
=&\frac{\alpha C^2}{BD}(\frac{1}{\frac{\sum_{i=1}^{N}\lambda_{i,min}^{(t)}}{B}} + \frac{1}{\frac{\sum_{i=1}^{N}\lambda_{i,min}^{(t)}}{B}-\frac{\alpha C^2}{BD}})
\end{aligned}
\end{equation}
Combing Eq. (\ref{eq:rd_bound1_1}) and Eq. (\ref{eq:rd_bound2_1}), we have:
\begin{equation}
\label{eq:rd_bound}
\begin{aligned}
&D_{\alpha}(\mathcal{M}_{\{D_i\cup D_{-i},\theta^{(t)}\}} \| \mathcal{M}_{\{D_i^{'}\cup D_{-i},\theta^{(t)}\}}) \\
\leq &\frac{2\alpha C^2}{D^2(\frac{\sum_{i=1}^{N}\lambda_{i,min}^{(t)}}{B}-\frac{\alpha C^2}{BD})}+\\
&\frac{\alpha C^2}{2BD(\alpha-1)}(\frac{1}{\frac{\sum_{i=1}^{N}\lambda_{i,min}^{(t)}}{B}} + \frac{1}{\frac{\sum_{i=1}^{N}\lambda_{i,min}^{(t)}}{B}-\frac{\alpha C^2}{BD}}) \\
\leq & (\frac{2\alpha BC^2}{D^2}+\frac{\alpha C^2}{(\alpha-1)D})\frac{1}{\sum_{i=1}^{N}\lambda_{i,min}^{(t)}-\frac{\alpha C^2}{D}}=\epsilon_i^{(t)}.
\end{aligned}
\end{equation}
Hence, the aggregated model update $x^{(t)}$ provides $(\alpha, \epsilon_i^{(t)})$-RDP for user $i$'s dataset $D_i$. 
Lastly, based on Lemma \ref{lemma1}, it can be derived that the aggregated model update $x^{(t)}$ provides $(\epsilon_i^{(t)}+\frac{\log(1/\delta)}{\alpha-1}, \delta)$-DP for user $i$'s dataset $D_i$.
\end{proof}


Theorem \ref{theorem1} indicates that the privacy bound $\epsilon$ in FL with SA mainly depends on two main factors: 1) the number of users participating in FL with SA, and 2) the minimal eigenvalues of the model update covariance matrix from each user. As we increase the number of users $N$, $\sum_{i=1}^{N}\lambda_{i,min}^{(t)}$ will increase and hence the $\epsilon$ will decay. Moreover, when $\lambda_{j,min}^{(t)}$, i.e., the minimal eigenvalues of any user $j$'s covariance matrix increase, $\sum_{i=1}^{N}\lambda_{i,min}^{(t)}$ will also increase and thus $\epsilon$ will decrease.  

However, it is worth noting that Theorem \ref{theorem1} relies on the non-singular covariance matrix assumption (i.e. Assumption \ref{assum3}), which may not easily hold in practice. Especially in applications using deep neural networks, the number of training data points is typically smaller than the number of model parameters (i.e. over-parameterization), and as a result the covariance matrix of the model update will be singular \cite{du2018gradient, allen2019convergence}. 
Motivated by this, we next consider the case where the sampling noise of each user is Gaussian and the gradient covariance matrix is singular.

\subsection{Gaussian sampling noise with singular gradient covariance matrix}
\label{subsec:gaussian_singular}
Before presenting our second theorem, we formally make the following assumption:
\begin{assump}[Non-singular covariance matrix in subspace] Assume that at step $t$, 
for any user $i$, the gradient $g_{i,j}^{(t)}$ calculated by any data point $d_{i,j}$ can be mapped into a subspace as $g_{i,j}^{(t)}=S^{(t)}g_{i,j}^{*(t)}$, where $S^{(t)}\in\mathbf{R}^{K\times K*}$ ($K> K^*)$ is a matrix consisting of $K^*$ orthogonal unit vector. We further assume that after being mapped to the subspace, the covariance matrix of gradient noise is $\Sigma_i^{*(t)}=\sum_{j=1}^{j=D}g_{i,j}^{*(t)}(g_{i,j}^{*(t)})^T$, which is non-singular and has non-zero minimal eigenvalue of $\lambda_{i,min}^{*(t)}>0$.
\label{assum4}
\end{assump}

\begin{thm}
Under Assumption \ref{assum1},\ref{assum2},\ref{assum4}, if $\frac{C^2}{D} < \sum_{i=1}^{N}\lambda_{i,min}^{(t)}$, then $\forall~\alpha\in (1, \frac{D\sum_{i=1}^{N}\lambda_{i,min}^{*(t)}}{C^2})$, the aggregated model update $x^{(t)}$ at step $t$ can provide $(\alpha,\epsilon_i^{*(t)})$-RDP and $(\epsilon_i^{*(t)}+\frac{\log(1/\delta)}{\alpha-1},\delta)$-DP for user $i$'s dataset $D_i$, where:
\begin{equation}
\label{eq:lambda1}
\epsilon_i^{*(t)} =
    (\frac{2\alpha BC^2}{D^2}+\frac{\alpha C^2}{(\alpha-1)D})\frac{1}{\sum_{i=1}^{N}\lambda_{i,min}^{*(t)}-\frac{\alpha C^2}{D}}.
\end{equation}
\label{theorem2}
\end{thm}
\begin{proof}
We define $G_{i}^{*(t)}=[g_{i,1}^{*(t)},...,g_{i,D}^{*(t)}]$. Based on Assumption \ref{assum4}, we have:

\begin{equation}
\label{eq:global_subspace}
\begin{aligned}
    x^{(t)}=&x_{i}^{(t)}+ x_{-i}^{(t)}=
    \sum_{i=1}^{N}G_{i}^{(t)}w_j = S^{(t)}\Big(\sum_{i=1}^{N}G_{i}^{*(t)}w_j\Big)\\
    =&S^{(t)}\Big(\sum_{i=1}^{N}\bar g_{i}^{*(t)}+ \sum_{i=1}^{N}\frac{1}{\sqrt{B}}L_{i}^{*(t)}v_i^*\Big)=S^{(t)}x^{*(t)},
\end{aligned}
\end{equation}
where $\bar g_{i}^{*(t)}=\frac{1}{D}\sum_{j=1}^{D}g_{i,j}^{*(t)}$ and $\Sigma_i^{*(t)}=U_{i}^{*(t)}\big(\lambda_{i}^{*(t)}\big)\big(U_{i}^{*(t)}\big)^T$, $L_{i}^{(t)}=U_{i}^{(t)}\big(\lambda_{i}^{(t)}\big)^{\frac{1}{2}}$, $v_i^*\sim \mathcal{N}(0, I_{K^*})$, and  $x^{*(t)}$ will be Gaussian with mean $\sum_{j=1}^{N}\bar g_{j}^{*(t)}$ and covariance matrix $\frac{\sum_{j=1}^{N}\Sigma_j^{*(t)}}{B}$.
Therefore, $\forall~x$ in the subspace of $S^{(t)}$ (i.e. $x=S^{(t)}x^*$), we have:
\begin{equation}
\begin{aligned}
        &\mathcal{M}_{\{D_i\cup D_{-i},\theta^{(t)}\}}(x) = \mathtt{Pr}[x_{i}^{(t)}+ x_{-i}^{(t)}=x|\{D_i\}_{i=1}^{i=N},\theta^{(t)}]\\
    = &\mathtt{Pr}[x^{*(t)}=x^*|\{D_i\}_{i=1}^{i=N},\theta^{(t)}].
\end{aligned}
\end{equation}
Following the proof of Theorem \ref{theorem1}, for any $D_i$ and $D_i^{'}$ (i.e. two instances of user $i$'s local dataset ) differing from one data point, we have:
\begin{equation}
\label{eq:rd_bound_2}
\begin{aligned}
&D_{\alpha}(\mathcal{M}_{\{D_i\cup D_{-i},\theta^{(t)}\}} \| \mathcal{M}_{\{D_i^{'}\cup D_{-i},\theta^{(t)}\}}) \\
\leq & (\frac{2\alpha BC^2}{D^2}+\frac{\alpha C^2}{(\alpha-1)D})\frac{1}{\sum_{i=1}^{N}\lambda_{i,min}^{*(t)}-\frac{\alpha C^2}{D}}=\epsilon_i^{*(t)}.
\end{aligned}
\end{equation}
Hence, the aggregated model update $x^{(t)}$ at step $t$ can provide $(\alpha,\epsilon_i^{*(t)})$-RDP and $(\epsilon_i^{*(t)}+\frac{\log(1/\delta)}{\alpha-1},\delta)$.
\end{proof}

Theorem \ref{theorem2} relies on Assumption \ref{assum4}. It states that for Gaussian gradient noise with singular covariance matrix, in order to provide DP guarantee for user $i$'s local dataset, any difference in aggregated model update caused by the change of one data point in must belong to the same subspace where the randomness in the aggregated model update belongs to. Otherwise, DP will be violated. The counter example in Fig.~\ref{fig:counter_example} violates DP since the model update of user 4 fails to satisfy Assumption \ref{assum4}. 

\noindent\textbf{Remark 2: Applicability of Assumption \ref{assum4} in practice.} Prior works have demonstrated that the gradient converges to a tiny subspace \cite{gur2018gradient} in the training process of overparameterized neural networks, and mapping gradient into a subspace (e.g. using PCA) can reduce the amount of noise added for the same level of DP \cite{abadi2016deep}. Hence, Assumption \ref{assum4} is a reasonable assumption for overparameterized neural network.
However, in order to verify whether Assumption \ref{assum4} holds, a centralized server is needed to access to each user's individual gradients and run SVD on them, which breaks the privacy guarantee provided by SA in FL (see Section \ref{subsection:sa_guarantee}). Therefore, in practice, a potential solution is that each user adds additional noise locally to make the covariance matrix non-singular, as discussed in Section \ref{sec:wf_noise}.




\section{Water-filling noise addition (WF-NA)  algorithm.}
\label{sec:wf_noise}
Till now, our theoretical results demonstrate the necessity of additional noise for DP guarantee in FL with SA when the covariance matrix of gradient noise is non-singular. In this section, we explore whether it is possible to leverage the inherent randomness inside aggregated model update to reduce the amount of additional noise required for DP. Specifically, we introduce a novel algorithm called Water-Filling noise addition (WF-NA), which lift the zero eigenvalues (and some small eigenvalues) in the covariance matrix of each user's gradient noise, in order to guarantee their non-singularity. For simplicity, we refer to this algorithm as WF-NA and describe its details below.

\subsection{Algorithm design}
The inputs of the WF-NA algorithm for each user $i$ include the local dataset $D_i$, the current global model parameters $\theta^{(t)}$, the number of users $N$ participating in FL, the clipping value $C$, mini-batch size $B$, the lower bound for minimal eigenvalue budget $\sigma^2$, and the relaxation item $\delta$. Note that without loss of generality, we assume that the input parameters are the same across users.

First, user $i$ utilizes its local dataset $D_i$ to compute the mean gradient and covariance matrix of gradient noise as $\mu_i^{(t)}$ and $\Sigma_{i}^{(t)}$ respectively. We describe this process for FedSGD separately below.

\noindent\textbf{Calculate covariance matrix.} For each data point in the local training dataset, user $i$ will calculate a gradient from this data point and clip its $L_2$ norm to make it upper bounded by constant $C$. 
After obtaining $|D_i|=D$ clipped gradient vectors, user $i$ uses them to calculate the mean $\mu_i^{(t)}=\frac{1}{D}\sum_{j=1}^{D}g_{i,j}^{(t)}$ and covariance matrix $\Sigma_{i}^{(t)}=\frac{1}{BD}\sum_{j=1}^{D}g_{i,j}^{(t)}(g_{i,j}^{(t)})^T$. 


\noindent\textbf{Run SVD and lower bound the smallest eigenvalue.} Next, user $i$ will run SVD on its estimate covariance matrix as $\Sigma_i^{(t)}=U_i^{(t)}\Lambda_i^{(t)}(U_i^{(t)})^{T}$, and clip the eigenvalues in $\Lambda_i^{(t)}$ such that all eigenvalues are lower bounded by $\sigma^2$. Define the updated diagonal matrix with bounded eigenvalues as $\Lambda_{i,+}^{(t)}$. The revised covariance matrix can be calculated as $\Sigma_{i,+}^{(t)}=U_i^{(t)}\Lambda_{i,+}^{(t)}(U_i^{(t)})^{T}$. Note that to make the SVD process efficient, we split the $d$-dimension gradient into $k$ parts, estimate the covariance matrix for each part, run SVD on these $k$ covariance matrices separately, and concatenate them into the final covariance matrix. By running SVD approximately, the time complexity of SVD will be reduced from $O(d^3)$ to $O(k\frac{d}{k}^3)=O(\frac{d^3}{k^2})$ (see similar approaches in \cite{liu2014additive, delattre2012blockwise,tamilmathi2022tensor}), without affecting the $\epsilon$ bound we have.

\noindent\textbf{Add WF noise.} Finally, each user adds Gaussian noise with covariance matrix $\Delta\Sigma_{i}^{(t)}=\Sigma_{i,+}^{(t)} - \Sigma_{i}^{(t)}$ into the local model update (i.e. $n\sim N(0,\Delta\Sigma_{i}^{(t)})$).

\begin{thm} [DP guarantees of Gaussian sampling noise + WF-NA algorithm]
\label{theorem_wf}
Given the input parameters of WF-NA algorithm, for any $\alpha\in(1,\frac{N\sigma^2D}{2C^2})$, the aggregated model update using WF-NA can provide $(\alpha,\epsilon^{(t)})$-RDP and $(\epsilon^{(t)}+\frac{\log(1/\delta)}{\alpha-1},\delta)$-DP guarantees to any user $i$'s local dataset, where the $\epsilon^{(t)}$ of each training round is given as follows:
\begin{equation}
\label{eq:wfdp}
\epsilon^{(t)} =
    (\frac{2\alpha C^2}{D^2}+\frac{2\alpha C^2}{(\alpha-1)BD})\frac{1}{N\sigma^2-\frac{2\alpha C^2}{BD}}.
\end{equation}
\end{thm}
\begin{proof}
We define the change of covariance matrix after applying WF-NA as: 
\begin{equation}
    \Sigma_{i,\Delta}^{(t)} = \Sigma_{i,+}^{(t)} - \Sigma_{i}^{(t)}.
\end{equation}
Since WF-NA only increases the eigenvalues of the original covariance matrix by maximally $\sigma^2$,   $\Sigma_{i,\Delta}^{(t)}$ will be semi-definite with maximal eigenvalue upper bounded by $\sigma^2$. Then, consider $\Sigma_{i,+}^{(t),1}$ and $\Sigma_{i,+}^{(t),2}$, which are two instances of $\Sigma_{i,+}^{(t)}$ by changing one data point in $D_i$, we have:
\begin{equation}
\begin{aligned}
      \Sigma_{i,+}^{(t),2}-\Sigma_{i,+}^{(t),1}=&\frac{1}{BD}(g_{i,1}^{(t),2}(g_{i,1}^{(t),2})^T-g_{i,1}^{(t),1}(g_{i,1}^{(t),1})^T) + \Delta\Sigma_{i,\Delta}^{(t)},  
\end{aligned}
\end{equation}
where $\Delta\Sigma_{i,\Delta}^{(t)} = \Sigma_{i,\Delta}^{(t),2}-\Sigma_{i,\Delta}^{(t),1}$. 

Since removing $g_{i,1}^{(t),1}(g_{i,1}^{(t),1})^T$ decreases the eigenvalues, which may cause some eigenvalues below smaller than $\sigma^2$, WF-NA needs to refill eigenvalues smaller than $\sigma^2$. Hence, we have $\Delta\Sigma_{i,\Delta}^{(t)}\preceq \frac{1}{BD}g_{i,1}^{(t),1}(g_{i,1}^{(t),1})^T$. This says, by changing one data point in user $i$'s local dataset, WF-NA will at most increase $\Sigma_{i,+}^{(t),1}$ by $\frac{1}{BD}g_{i,1}^{(t),1}(g_{i,1}^{(t),1})^T$ to refill small eigenvalues. 
Moreover, since adding $g_{i,1}^{(t),2}(g_{i,1}^{(t),2})^T$ increases the eigenvalues, WF-NA may lift eigenvalues smaller than $\sigma^2$ less. Hence, we have $\Delta\Sigma_{i,\Delta}^{(t)}\succeq -\frac{1}{BD}g_{i,1}^{(t),2}(g_{i,1}^{(t),2})^T$. This says, by changing one data point in user $i$'s local dataset, WF-NA will at most reduce $\Sigma_{i,+}^{(t),1}$ by $\frac{1}{BD}g_{i,1}^{(t),2}(g_{i,1}^{(t),2})^T$. Formally, we have:
\begin{equation}
\begin{aligned}
      -\frac{1}{BD}g_{i,1}^{(t),2}(g_{i,1}^{(t),2})^T \preceq \Delta\Sigma_{i,\Delta}^{(t)} \preceq \frac{1}{BD}g_{i,1}^{(t),1}(g_{i,1}^{(t),1})^T
\end{aligned}
\end{equation}

Since $||g_{i,1}^{(t),l}||_2\leq C$, we can derive that $\forall~x\in\mathbf{R}^K$,
\begin{equation}
\begin{aligned}
      &x^T(\Sigma_{i,+}^{(t),2}-\Sigma_{i,+}^{(t),1})x\\
      \geq &-x^T(\frac{1}{BD}g_{i,1}^{(t),1}(g_{i,1}^{(t),1})^T+\frac{1}{BD}g_{i,1}^{(t),2}(g_{i,1}^{(t),2})^T)x\geq -\frac{2C^2}{BD}x^Tx.  
\end{aligned}
\end{equation}
Therefore, $\forall~x\in\mathbf{R}^K$, 
\begin{equation}
\begin{aligned}
&x^T\Sigma_{\alpha,+} x=x^T((1-\alpha)\Sigma_{1,+} + \alpha\Sigma_{2,+})x \\
=&x^T\Sigma_1 x + x^T\frac{\alpha}{B}(\Sigma_{i,+}^{(t),2}-\Sigma_{i,+}^{(t),1})x\\
\geq & (N\sigma^2-\frac{2\alpha C^2}{BD})x^Tx.
\end{aligned}
\end{equation}
Moreover, we can derive
\begin{equation}
\Sigma_{2,+}\preceq \Sigma_{1,+} + \frac{1}{BD}g_{i,1}^{(t),2}(g_{i,1}^{(t),2})^T +  \frac{1}{BD}g_{i,1}^{(t),1}(g_{i,1}^{(t),1})^T,
\end{equation}
and 
\begin{equation}
\Sigma_{1,+}\preceq \Sigma_{\alpha,+} + \frac{\alpha}{BD}g_{i,1}^{(t),2}(g_{i,1}^{(t),2})^T +  \frac{\alpha}{BD}g_{i,1}^{(t),1}(g_{i,1}^{(t),1})^T.
\end{equation}
Following the proof of Theorem \ref{theorem1}, we derive:
\begin{equation}
\label{eq:rd_bound_3}
\begin{aligned}
&D_{\alpha}(\mathcal{M}_{\{D_i^1\cup D_{-i},\theta^{(t)}\}} \| \mathcal{M}_{\{D_i^2\cup D_{-i},\theta^{(t)}\}}) \\
\leq & \frac{\alpha}{2}(\Delta \bar g_{i}^{(t)})^T \Sigma_{\alpha,+}^{-1}\Delta \bar g_{i}^{(t)}-\frac{1}{2(\alpha - 1)}\ln\frac{|\Sigma_{\alpha,+}|}{|\Sigma_{1,+}|^{1-\alpha}|\Sigma_{2,+}|^{\alpha}}.\\
\leq & \frac{2\alpha C^2}{D^2}\frac{1}{N\sigma^2-\frac{2\alpha C^2}{BD}} + \frac{\alpha C^2}{(\alpha-1)BD}(\frac{1}{N\sigma^2}+\frac{1}{N\sigma^2-\frac{2\alpha C^2}{BD}})\\
= & (\frac{2\alpha C^2}{D^2}+\frac{2\alpha C^2}{(\alpha-1)BD})\frac{1}{N\sigma^2-\frac{2\alpha C^2}{D}}=\epsilon^{(t)}.
\end{aligned}
\end{equation}
Hence, the aggregated model update using WF-NA can provide $(\alpha,\epsilon^{(t)})$-RDP and $(\epsilon^{(t)}+\frac{\log(1/\delta)}{\alpha-1},\delta)$-DP guarantees to any user $i$'s local dataset.
\end{proof}
    

\subsection{Discussion about WF-NA}
\label{subsec:intuitive_cmp}

\noindent\textbf{Intuitive comparison of WF noise with isotropic noise.}
\label{subsec:wf_comp}
Figure \ref{fig:noise_dp} compares WF noise with isotropic noise, where the $x$-axis represents the index of each eigenvalue and the $y$-axis represents the value of each eigenvalue. The noise contained in the aggregated model update is randomly distributed in the directions of all eigenvectors, and each eigenvalue measures the variance of the noise in each eigenvector direction. When the eigenvalue is larger than $\sigma^2$ (i.e. the amount of noise we need for achieving our privacy target), then the aggregated model update can already provide enough privacy protection, without needing additional noise. Since WF-NA leverages the noise contained in the aggregated model update, by only adding noise to the eigenvector directions where the eigenvalues are smaller than  $\sigma^2$, it can add less noise compared with adding isotropic noise, as shown in Figure \ref{fig:wf_dp} and \ref{fig:ldp}. 

\begin{figure}[t]
  \centering
  \begin{subfigure}{0.48\linewidth}
    \includegraphics[width=0.85\textwidth]{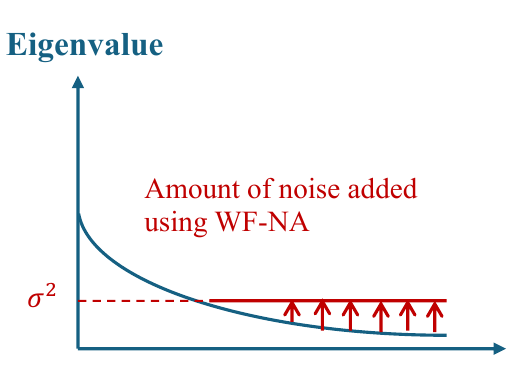}
    \caption{WF noise.}
    \label{fig:wf_dp}
  \end{subfigure}
  \hfill
  \begin{subfigure}{0.48\linewidth}
    \includegraphics[width=0.85\textwidth]{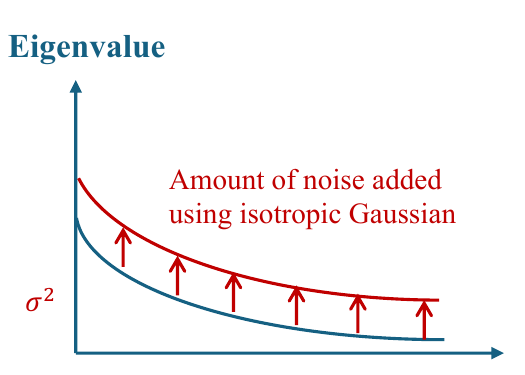}
    \caption{Isotropic noise.}
    \label{fig:ldp}
  \end{subfigure}
  \caption{Comparison of WF noise and isotropic noise.}
  \label{fig:noise_dp}
\end{figure}

\begin{table*}[!h]
\centering
\begin{tabular}{p{1.2in}<{\centering}p{1.2in}<{\centering}|p{3.8in}<{\centering}}
\hline
Method for local model update calculation & Additional noise & $(\epsilon,\delta)$-DP guarantee\\ \hline
SGD + IID gradient sampling  & Isotropic Gaussian & $\epsilon=\log(1+\frac{B}{D}(e^{\epsilon'}-1))$, where $\epsilon'=\displaystyle\min_{\alpha>1}\frac{\log(1/\delta)}{\alpha-1}+\frac{2\alpha C^2}{B^2}\cdot\frac{1}{N\sigma^2}$ (\cite{abadi2016deep,bun2018composable})\\
SGD + IID gradient sampling & WF-NA & No DP guarantee \\\hline
GD without sampling & Isotropic Gaussian & $\epsilon =\displaystyle\min_{\alpha>1}\frac{\log(1/\delta)}{\alpha-1}+\frac{2\alpha C^2}{D^2}\cdot\frac{1}{N\sigma^2}$ (\cite{abadi2016deep}) \\
GD without sampling & WF-NA & No DP guarantee\\\hline
SGD + Gaussian gradient sampling & Isotropic Gaussian & $\epsilon=\displaystyle\min_{\alpha>1}\frac{\log(1/\delta)}{\alpha-1}+(\frac{2\alpha C^2}{D^2}+\frac{\alpha C^2}{(\alpha-1)BD})\frac{1}{N\sigma^2-\frac{\alpha C^2}{BD}}$\\
SGD + Gaussian gradient sampling & WF-NA & $\epsilon=\displaystyle\min_{\alpha>1}\frac{\log(1/\delta)}{\alpha-1}+(\frac{2\alpha C^2}{D^2}+\frac{\alpha C^2}{(\alpha-1)BD})\frac{1}{N\sigma^2-\frac{\alpha C^2}{BD}}$\\\hline

\end{tabular}
\caption{Comparison of different DP mechanisms in FedSGD. Note that for SGD + IID gradient sampling, the sampling ratio is $\frac{B}{D}$, which will offer a privacy amplification with ratio $\frac{B}{D}$ approximately.}
\label{tab:tab_dp_cmp}
\end{table*}

\noindent\textbf{Comparison of different DP mechanisms.} In Table \ref{tab:tab_dp_cmp}, we summarize the DP guarantees of different sampling noise with WF-NA and compare them with the DP guarantees provided by traditional DP mechanisms using isotropic Gaussian noise.  It is worth noting that WF-NA itself is not enough to guarantee DP, since it only guarantees that the covariance matrix of gradient noise is non-singular. In addition, it needs to be combined with appropriate sampling noise in order to achieve DP guarantee (e.g. Gaussian sampling noise). 

As shown in Table \ref{tab:tab_dp_cmp}, when SGD with IID sampling or Gradient Descent (GD) without gradient sampling is used, the usage of WF-NA cannot provide any DP guarantee, since WF-NA does not guarantee the necessary condition in Theorem \ref{theorem1} (see Section \ref{subsec:gaussian_non_singular} for details). Moreover, compared with mechanisms using SGD, the usage of GD and isotropic Gaussian noise injects the least amount of noise into the local model updates, since GD does not introduce any sampling noise caused by gradient sampling, but it does have larger cost than SGD with IID sampling since it uses all data points. Last, when SGD is used, while WF-NA injects less noise compared with the addition of isotropic Gaussian, Gaussian gradient sampling may inject more sampling noise compared with IID gradient sampling. Hence, the total amount of noise added by Gaussian gradient sampling + WF-NA may not be smaller than the total amount of noise added by IID gradient sampling + isotropic Gaussian. In the next subsection, we empirically compare the performance of these mechanisms for the same level of privacy.

\subsection{Comparison of DP mechanisms}
\label{subsec:simulation}
In this subsection, we consider the tasks of training both a linear model and a Convolutional Neural Network (CNN) model on MNIST dataset with 50 users. We compare the accuracy of these models when different DP mechanisms are employed during the training process, to verify our theoretical analysis in Section \ref{subsec:wf_comp}. Specifically, we compare three mechanisms, described as follows.\\
\noindent\textbf{SGD-DP:} It uses SGD with IID sampling and the addition of isotropic Gaussian (i.e. adding noise in the whole gradient space). \\
\noindent\textbf{GSGD-DP:} It uses SGD with Gaussian sampling and WF-NA (i.e. adding noise to the ``unprotected'' gradient subspace). \\
\noindent\textbf{GD-DP:} It uses GD without any sampling noise and the addition of isotropic Gaussian (i.e. adding noise in the whole gradient space).

As demonstrated in Figure \ref{fig:mnist}, GSGD-DP has significantly better accuracy than SGD-DP. However, both of them exhibit lower accuracy compared with GD-DP. In terms of training efficiency, SGD-DP is better than GD-DP and GSGD-DP, since it merely requires users to sample a mini-batch of $B$ data points and take the average of gradients calculated from these data points as model update. Moreover, GSGD-DP is the least efficient compared to GD-DP and SGD-DP, due to the necessity for WF-NA to run SVD on the covariance matrix of gradient noise (see Section \ref{sec:wf_noise} for details). Therefore, we conclude that GD-DP is preferable for maximizing model accuracy without considering training efficiency. Conversely, SGD-DP is recommended when training efficiency is prioritized. Note that when the sampling ratio in SGD (i.e. $\frac{B}{D}$) converges to 1, the performance of SGD-DP will also converge to GD-DP. This says, the key difference between SGD-DP and GD-DP is how they trade between training efficiency and model accuracy for the same level of privacy. 
As indicated in \cite{cheng2019towards}, the mini-batch size $B$ in SGD-DP can be adjusted to trade between training efficiency and training accuracy. It is noteworthy that while GSGD-DP appears to be impractical due to the requirement to run SVD, investigating how to add minimal noise only to the model update subspace requiring additional noise for a DP guarantee is a compelling research direction. We leave it as future research to explore whether this idea can be practically applied (see Section \ref{sec:discussion} for details).

\begin{figure}
  \centering
  \begin{subfigure}{0.48\linewidth}
    \includegraphics[width=0.9\textwidth]{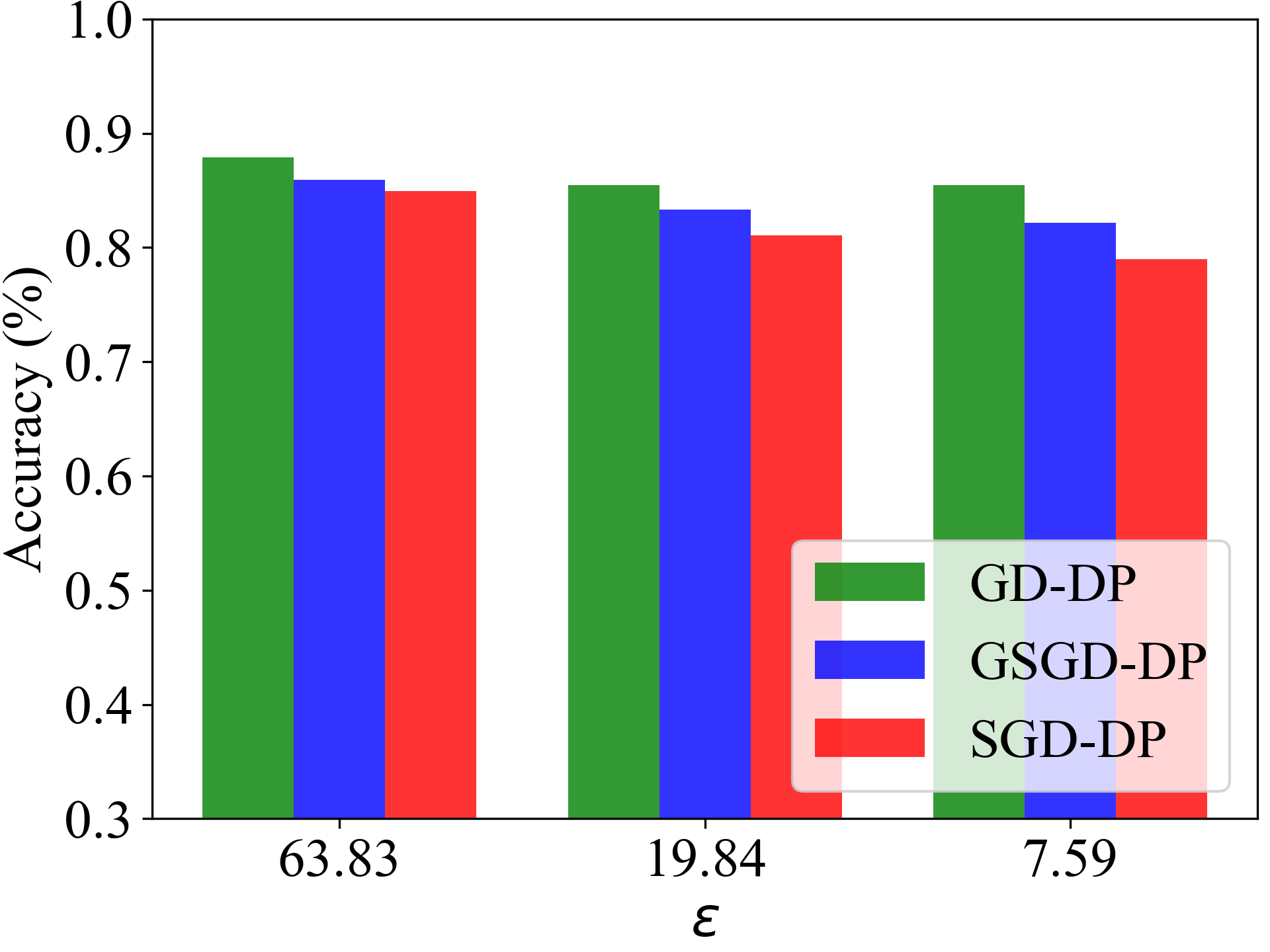}
    \caption{Linear model.}
    \label{fig:mnist_linear}
  \end{subfigure}
  \hfill
  \begin{subfigure}{0.48\linewidth}
    \includegraphics[width=0.9\textwidth]{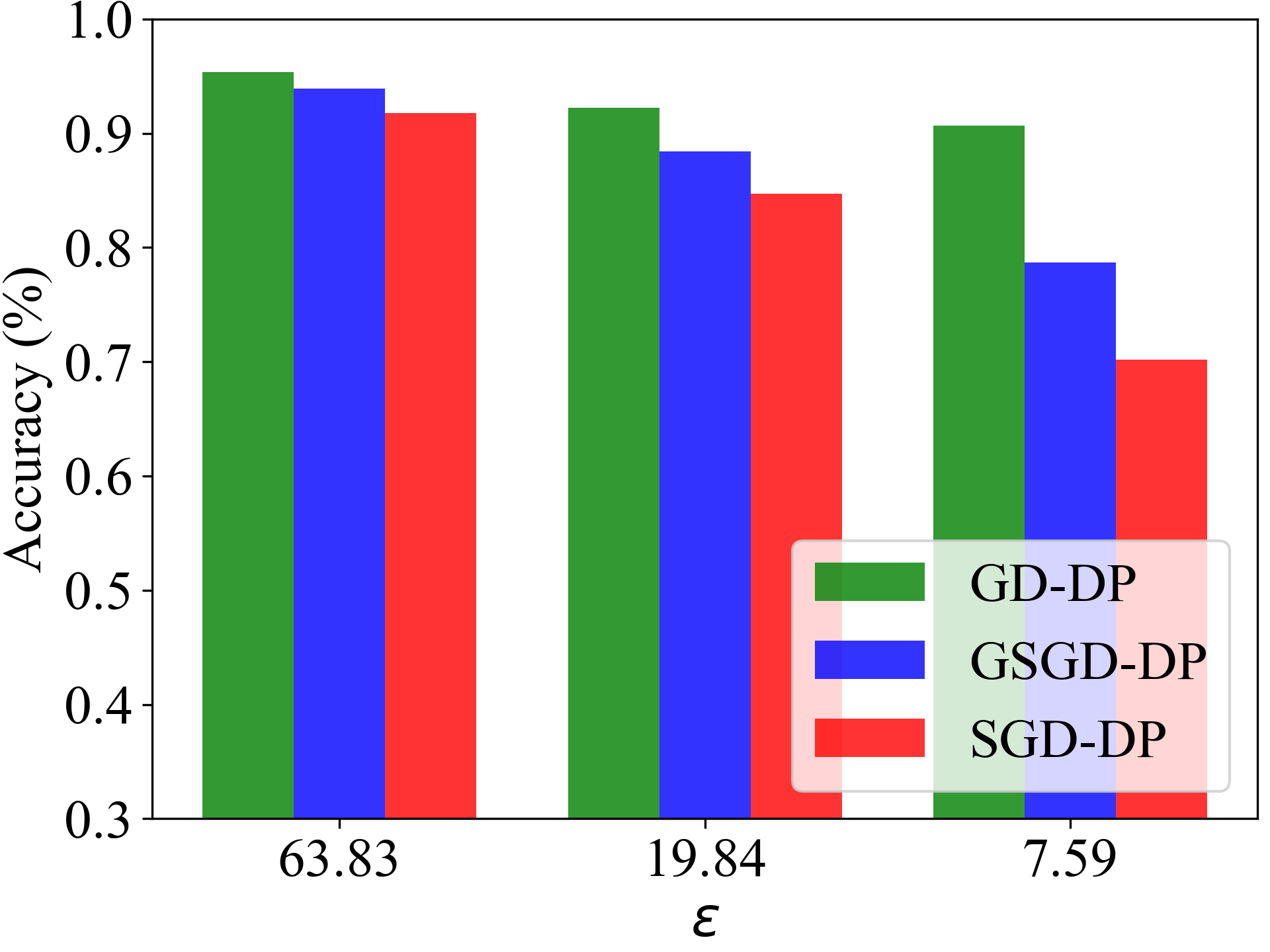}
    \caption{CNN model.}
    \label{fig:mnist_cnn}
  \end{subfigure}
  \caption{Comparison of different DP mechanisms on MNIST dataset. Note that we consider 50 users participating in FL. The training epoch is set as 100, the mini-batch size $B$ is 32, the clipped value $C$ is set as 10, and we consider $\delta=10^{-4}$. We report the accumulative privacy across all training epochs by using the composition theorem in \cite{kairouz2015composition}.}
  \label{fig:mnist}
\end{figure}

\section{Discussion and Future Work}
\label{sec:discussion}
\noindent\textbf{Generalization to other FL protocols.} While our theoretical results are derived based on FedSGD, they can be generalized to other FL algorithms such as FedAvg \cite{mcmahan2017communication}, by only varying how we generate the model update (i.e. Assumption \ref{assum1}). As an example, in FedAvg, at training round $t$ each user conducts multiple local training steps (i.e. a random process) to derive their local model update $x_{i}^{(t)}$, which will also be a random vector with some mean and covariance matrix. Therefore, after changing Assumption \ref{assum1} w.r.t. how we generate $x_{i}^{(t)}$, Theorem \ref{theorem1}-\ref{theorem2} can be applied.

\noindent\textbf{Leveraging inherent randomness in model updates.} As explored in Section \ref{sec:wf_noise}, when the inherent randomness in the aggregated model update can be leveraged, the amount of additional noise required for achieving DP guarantee can be reduced, since we only need to add additional noise to the subspace where the randomness in aggregated model updates cannot offer enough privacy protection (i.e. the `unprotected'' subspace). 
Despite its potential, several challenges exist in the design of a practical mechanism to leverage the inherent randomness in aggregated model updates to reduce the amount of noise needed for DP.

First, without making any assumptions about the distribution of each individual user's model update, we are unable to quantify and analyze the distribution of the inherent randomness in the aggregated model update. Hence, the closed-form of $\epsilon$ cannot be derived through theoretical analysis. A potential solution is to approximate the original aggregated model update whose distribution is unknown with a model update satisfying a common distribution. For example, as demonstrated in Section \ref{sec:theory}, Gaussian gradient sampling noise can be used to make the aggregated model update be Gaussian. While such gradient approximation makes it analytically feasible to derive the bound of $\epsilon$ in DP, it actually introduces additional noise into the raw aggregated model update (see Section \ref{sec:wf_noise} for details). Therefore, compared with traditional DP mechanisms that add noise directly to the original gradient, whether novel gradient approximation approaches can be developed to improve the privacy-utility trade-offs during DP noise addition deserves further research.

Second, the computational overhead of quantifying and leveraging the inherent randomness in the model update can be large. As demonstrated in Section \ref{sec:theory} and Section \ref{sec:wf_noise}, quantifying the ``unprotected'' subspace needs SVD, which will significantly increase the runtime overhead in practice. Hence, a practical direction of future work is to explore how to efficiently quantify and use the inherent randomness in the aggregated model update.  

Third, to ensure DP guarantee with meaningful $\epsilon$ for each individual user, users need to clip their gradients such that the distribution shift of the aggregated model update caused by the change of one data point locally can be bounded and hence a meaningful bound for $\epsilon$ can be guaranteed. On the other hand, if each user
clips their individual gradient, the aggregated model update will be bounded and hence cannot be used as unbounded noise with respect to other
users. In this case, when applying DP mechanisms with unbounded noise (e.g. Gaussian mechanism, Laplacian mechanism), the bounded inherent randomness in the aggregated model update may not help to reduce the variance of unbounded noise needed for DP guarantee. To address this, potential strategies include the adoption of DP mechanisms that utilize bounded noise distributions (e.g \cite{holohan2018bounded,dagan2022bounded,liu2018generalized}). Another approach involves discretizing the gradient, enabling the addition of bounded discrete noise into the model update \cite{kairouz2021distributed}.

\vspace{-0.05in}
\section{Related work}
\label{sec:related}





Secure aggregation (SA) has been introduced to address information leakage from model updates in the context of federated learning (FL). The  state-of-the-art uses additive masking to protect the privacy of individual models \cite{cc,secagg_bell2020secure,secagg_so2021securing,secagg_kadhe2020fastsecagg,zhao2021information,so2021lightsecagg,so2021turbo}. SecAg \cite{cc} was the first practical  protocol proposed for   FL with secure model aggregation that is resilient for both user failures or dropouts and  collusion between users as well as the server, and it is based on  pairwise secret keys to be generated between each pair of users for masking the model updates. 
The cost of constructing and sharing  these masks is quite large though, scaling with $O(N^2)$ where $N$ corresponds to the number of users.  Recent works have managed to reduce the complexity of SecAg to $O(N\log N)$, see SecAg+ \cite{secagg_bell2020secure} and TruboAggregate  \cite{so2021turbo}. SecAg+ leverages a sparse random graph  where each user jointly encodes its model update with only a subset of user,  while TruboAggregate leverages both  sequential training over groups of rings and lagrange coded computing \cite{yu2019lagrange}. 
Further overhead reduction has been recently achieved by LightSecAg \cite{so2021lightsecagg} which uses  private MDS codes. The reduction in complexity is based on using one-shot aggregate-mask reconstruction of the surviving users. Other secure aggregation protocols proposed to reduce the computation/communication  complexity of SecAg include \cite{secagg_kadhe2020fastsecagg,zhao2021information}.

A different direction to SA to provide privacy in the context of FL is to use differential privacy (DP) by adding some noise to either the data or the model updates. Differential privacy (DP) provides robust mathematical privacy guarantees by ensuring that the individual contribution of a client does not have a significant impact on the final learning result. In a trusted server setting, earlier works such as DP-SGD and DP-FedAvg~\cite{bassily2014private,abadi2016deep, mcmahan2017communication} have focused on applying DP centrally in FL, where the server is trusted with individual model updates before implementing the deferentially private mechanism.
An alternative approach, local differential privacy (LDP), explores similar guarantees where the server is not to be trusted. In LDP the updates are perturbed at the client-side before it is collected by the server in the clear. The LDP model provides stronger privacy guarantees as it does not require a fully trusted aggregator
but it suffers from a poor privacy-utility trade-off. To improve privacy-utility trade-off of LDP, prior works have proposed to use parameter shuffling (\cite{girgis2021shuffled,sun2020ldp}) and different FL protocols (\cite{noble2022differentially}). However, these approaches do not consider the usage of SA and hence lead to a significant utility drop \cite{kairouz2021distributed}.

Although FL systems may still leak information even when SA is utilized, there is an opportunity to apply distributed DP mechanisms which may distributed the added noise among the users thus improving utility, see, for example, recent efforts along these lines~\cite{agarwal2018cpsgd, kairouz2021distributed, agarwal2021skellam, chen2022poisson, yang2023privatefl}. Specifically, in~\cite{agarwal2018cpsgd, kairouz2021distributed, agarwal2021skellam}, differential privacy is achieved by adding novel discrete noise
that asymptotically converges to
a normal distribution. In~\cite{chen2022poisson}, the authors propose a Poisson Binomial Mechanism for distributed differential privacy that results in an unbiased aggregate update at the server. In~\cite{yang2023privatefl}, the authors design a novel transformer layer deployed locally for each user, which effectively reduces the accuracy drop caused by added noise in SA. 
However, these approaches do not take advantage of the inherent randomness in FL from updates of other participating clients in the system, which is the focus of this work. While there are a few works \cite{wang2015privacy,dong2022privacy,carlini2022no} which leverage the inherent randomness from data sampling in centralized learning to protect DP guarantee, they do not quantify to inherent randomness from model updates of other users in FL with SA.

Very recently, the authors in \cite{elkordy2022much} have used mutual information to show that the inherent randomness from updates from other users may offer some privacy protection in the context of FL with SA. In this work, we focus on whether FL with SA can provide differential privacy guarantees, which is a much stronger privacy guarantee.
\vspace{-0.05in}
\section{Conclusions}
In this paper, we formally analyze the conditions under which FL with SA can provide DP guarantees without additional noise. We then demonstrate that the impracticability of these conditions, especially for over-parameterized neural networks. Hence, additional noise is necessary to achieve DP guarantee in FL with SA in practice. Future work could explore  how to leverage the inherent randomness inside aggregated model update to reduce the amount of additional noise required for DP guarantee.


\bibliographystyle{ACM-Reference-Format}
\bibliography{reference}

\end{document}